\documentclass{article}

\usepackage[utf8]{inputenc}
\usepackage[round]{natbib}
\usepackage[english]{babel}
\usepackage{amsfonts}
\usepackage{amsmath}
\usepackage{amsthm}
\usepackage{amssymb}
\usepackage{dsfont}
\usepackage{url}
\usepackage{hyperref}
\usepackage{todonotes}
\usepackage{tikz-cd}
\usepackage{mathpartir}
\usepackage{mathtools}
\usepackage{enumitem}
\usepackage{color}
\usepackage{stmaryrd}
\usepackage{listings}
\mathtoolsset{showonlyrefs}

\newtheorem{theorem}{Theorem}
\newtheorem{corollary}{Corollary}
\newtheorem{lemma}{Lemma}
\newtheorem{definition}{Definition}
\newtheorem{definition-proposition}{Definition and Proposition}
\newtheorem{proposition}{Proposition}
\newtheorem{assumption}{Assumption}

\begin{document}
\title{Synthetic topology in Homotopy Type Theory for probabilistic
  programming}

\author{Martin E.\@ Bidlingmaier \and Florian Faissole \and Bas Spitters}
\maketitle

\begin{abstract}
The ALEA Coq library formalizes measure theory based on a variant of the Giry monad on the category of sets.
This enables the interpretation of a probabilistic programming language with primitives for sampling from discrete distributions.
However, continuous distributions have to be discretized because the corresponding measures cannot be defined on all subsets of their carriers.

This paper proposes the use of synthetic topology to model continuous distributions for probabilistic computations in type theory.
We study the initial $\sigma$-frame and the corresponding induced topology on arbitrary sets.
Based on these intrinsic topologies we define valuations and lower integrals on sets, and prove versions of the Riesz and Fubini theorems.
We then show how the Lebesgue valuation, and hence continuous distributions, can be constructed.
\end{abstract}

\section{Introduction}
Monads on Cartesian closed categories are a semantics for a large class of effectful functional programming languages \citep{moggi1991notions}.
The ALEA Coq library \citep{ALEA} provides an interpretation of $\mathcal{R}$ml, a functional programming language with primitives for random choice, by constructing a version of the Giry monad~\citep{giry1982categorical} on the category of Coq's types.
Giry monads generally assign to a suitable class of spaces their spaces of measures or valuations, and in ALEA's case it is the class of discrete spaces.
ALEA's monad is suitable for embedding programming languages with discrete sampling constructs into the ambient logic of Coq, as for example in applications to cryptography \citep{Bguelin2010FormalCO}.
But continuous distributions are essential in statistics, machine learning and differential privacy, and these distributions have to be discretized in ALEA because they cannot be defined on discrete spaces.
For example, the Lebesgue measure is only defined on Borel sets, and hence is not directly definable in ALEA.

We propose the use of \emph{synthetic topology} \citep{lesnikphd} as a principled way of resolving the problem of continuous distributions.
In synthetic topology, one works in constructive (in our case even predicative) mathematics to which one adds axioms that make sets behave much like topological spaces.
The precise mathematical foundation we have in mind is \emph{Homotopy Type Theory (HoTT)} as it is used in modern proof assistants, on top of which we assume the necessary axioms of synthetic topology.

HoTT has a number of advantages over standard intensional type theory, even when one is only interested in sets, i.e.\@ types with trivial higher structure.
ALEA can only prove its version of the Giry monad to adhere to the
monad laws pointwise and resorts to setoids because neither function
extensionality nor quotients are part of standard Coq.
This is not a problem in HoTT, where function extensionality is provable and quotients of sets can be constructed as a special case of higher inductive types.
We refer to Section \ref{sec:preliminaries} for a detailed discussion of our mathematical foundations.

Our main contribution is a development of the theory of valuations (which play the role of measures) and lower integrals on sets in synthetic topology.
We show that a version of the Riesz theorem holds in this setting:
Valuations are in one-to-one correspondence with lower integrals.
This is then used to define a Giry monad $\mathfrak{G}$ on the category of sets in terms of the continuation monad, and we prove a version of the Fubini theorem.
Assuming the \emph{metrizability} of the real numbers $\mathbb{R}$, which asserts that the intrinsic topology on the set $\mathbb{R}$ agrees with the metric topology, we then define the Lebesgue valuation as an element of $\mathfrak{G}(\mathbb{R})$.
Finally, we obtain an interpretation of $\mathcal{R}$ml, a call-by-value PCF with probabilistic effects, via the restriction of the Giry monad $\mathfrak{G}$ to sub-probability valuations.

In non-classical measure theory (which is required because the
metrizability of $\mathbb{R}$ is contradictory with classical logic),
the Dedekind or Cauchy real numbers have to be replaced by the lower
reals $\mathbb{R}_l$ because the former are not closed under
enumerable suprema.
A lower real is a lower closed rounded inhabited subset of $\mathbb{Q}$, and in synthetic topology it is natural to require that this subset is furthermore an open subset.
An analogous construction for Dedekind reals in synthetic topology is studied by \citet{lesnikphd} in great generality. 
The HoTT book \citep{hottbook} also proposes this in the special case of $S$ equal to the initial $\sigma$-frame, and a formalization inspired by Coq's Math Classes \citep{math-classes} using the HoTT library \citep{HoTTlibrary} has been carried out by \citet{Gilbert}.
We develop the theory of lower reals valued in the initial $\sigma$-frame and construct an isomorphism $\mathbb{R}_l \cong \mathbb{Q}_\omega$ with the $\omega$-cpo completion of the rationals $\mathbb{Q}$.

The initial $\sigma$-frame is itself the $\omega$-cpo completion of the partial order of booleans $\bot \leq \top$, or, equivalently, the pointed $\omega$-cpo completion of the unit set $1 = \{ * \}$.
Pointed $\omega$-cpo completions of sets are studied by \citet{Altenkirch} in HoTT using quotient inductive inductive types \citep{altenkirch2018quotient}.
We explain how their construction can be adapted to $\omega$-cpo completions of preorders with respect to covers.
This generality is needed to define $\omega$-cpo completions of the rationals and the definition of a formal $\sigma$-frame of opens in the Dedekind reals $\mathbb{R}$.

In concurrent work with our initial work on this
topic \citep{faissole:hal-01654459}, \citet{Huang:thesis}
developed the semantics of a probabilistic programming language
targeted at machine learning with semantics in topological domains. 
Meanwhile, \citet{HuangMorrisettSpitters} have connected the two approaches
by showing that the
interpretation of a valuation in the internal logic of the
K2-realizability topos indeed gives the notion of valuation on
topological domains as defined by \citet{Huang:thesis}.

Some of the results presented in this paper are formalized\footnote{\url{https://github.com/FFaissole/Valuations/tree/d06d2c8c9cce3ddf6137}} in Coq on top of the HoTT library.
The core of the formalization consists of a proof of the dominance property of the Sierpinski space (Theorem \ref{th:dominance}), most of our discussion of the lower reals (Section \ref{sec:lower-reals}), and the definition of the Giry monad (Definition and Proposition \ref{def-prop:giry}).

The paper is structured as follows.
Section \ref{sec:preliminaries} contains some of the order-theoretic preliminaries and notation used throughout the paper.
Section \ref{sec:presentations} discusses the construction and properties of $\omega$-cpo completions.
Section \ref{sec:synth-top} studies the initial $\sigma$-frame as a set of truth values in synthetic topology.
Section \ref{sec:lower-reals} constructs the lower reals and contains a proof of their universal property (Theorem \ref{th:lower-reals-as-completion}).
Section \ref{sec:valuations-integrals} defines valuations and integrals and proves their equivalence (Theorem \ref{th:riesz}, a variant of the Riesz theorem).
Section \ref{sec:giry} constructs the Giry monad and proves a variant of the Fubini theorem (Theorem \ref{th:fubini}).
Section \ref{sec:lebesgue} discusses the metrizability of $\mathbb{R}$ and constructs the Lebesgue valuation.
Section \ref{sec:interpretation} provides an interpretation of $\mathcal{R}$ml based on the Giry monad that can account for continuous distributions.
Section \ref{sec:conclusion} concludes.

\section{Preliminaries}
\label{sec:preliminaries}

\paragraph{Logical foundations}

Our logical foundation is predicative and constructive mathematics.
\emph{Constructivity} means that we do not assume classical principles such as the lemma of the excluded middle, and that we do not assume the existence of choice functions.
\emph{Predicativity} means that we do not use the powerset construction, i.e.\@ we do not assume that there are sets $\mathcal{P}(A) = \{ B \mid B \subseteq A \}$.

The concrete system that we have in mind is \emph{Homotopy Type Theory (HoTT)}, and specifically its theory of \emph{(homotopy) sets}, i.e.\@ types with trivial higher structure.
\citet{RijkeSpitters} prove that the category of sets in HoTT form a $\Pi W$-pretopos, which is the category theoretic description of our logical foundation.

HoTT's inductive types allow the construction of effective quotients.
Effectivity means that the principle of unique choice holds:
If for a binary relation $R \subseteq X \times Y$ we have that for all $x \in X$ there exists a unique $y \in Y$ such that $R(x, y)$, then there exists a function $f_R : X \rightarrow Y$ such that $R(x, f_R(x))$ for all $x$.
Furthermore, functions in HoTT satisfy function extensionality:
If two functions agree pointwise, then the two functions are equal.

These two principles are fairly unusual for systems based on type theory, and they are not provable in Agda or Coq.
While some of HoTT's higher structure, in particular the univalence axiom, can make it more convenient to work even with hsets, higher types are not strictly required for our work.
We are thus optimistic that our work could in principle be formalized in systems such as XTT \citep{sterling2020cubical} or OTT \citep{ott} that promise to extract HoTT's well-behaved logic of hsets while discarding higher principles such as the univalence axiom for hsets.
However, not all quotients in XTT and OTT are well-behaved in the sense that they satisfy the principle of unique choice, hence our reliance on it might obstruct such a formalization project.

Predicative foundations reject the notion of a subobject classifier (i.e. a set $\Omega$ of truth values such that every subset of a set $X$ corresponds to a map $X \rightarrow \Omega$), but they permit universes of small, bounded sets.
We thus assume a countable hierarchy of universes $\mathcal{U}_0 \subseteq \mathcal{U}_1 \subseteq \dots$ of small sets.
Universes allow the definition of sets of small proposition $\Omega_0 \subseteq \Omega_1 \subseteq \dots$ by restriction to sets with at most one element.
We thus obtain small powersets $\mathcal{P}_i(X)$ as sets of functions $X \rightarrow \Omega_i$.
The bookkeeping of the current universe/subset level $i$ is essentially trivial; we will thus simply write $\Omega$ to mean the set $\Omega_i$ for a fixed level $i$ where confusion is unlikely.

While all existing systems that implement our intended logical foundations are based on type theory, we stick to the usual set theoretic notation in this paper.
We emphasize that the difference between our set theoretic notation and type theory is only superficial:
For example, when we write $X \subseteq Y$, we mean that there is an evident injective coercion from $X$ to $Y$.
A set comprehension $\{ x \in X \mid \phi(x) \}$ corresponds to the dependent sum type $\Sigma_{x \in X} \phi(x)$ and is used only when $\phi(x)$ is a proposition for all $x$.
We adopt the convention that the phrase ``there exists'' refers to a proof-relevant, i.e. untruncated statement; when we mean the proof irrelevant notion we say that something ``\emph{merely}'' exists.
Similarly, we write $X \rightarrow Y$ for the set of functions from $X$ to $Y$ here; elsewhere, this set is often denoted by $Y^X$.

In addition to the logical foundations based on HoTT's hsets, we require two additional principles to work with synthetic topology.
First, we assume the existence of \emph{free $\omega$-cpo completions} (Assumption \ref{ass:free-omega-cpo}).
As explained in Section \ref{sec:presentations}, this is a fairly weak assumption; it follows from a number of other common axioms (impredicativity, HoTT's quotient-inductive-inductive types or the axiom of countable choice).

Second and more invasive, we assume that the set of real numbers is \emph{metrizable} (Assumption \ref{ass:reals-metrizable}).
The metrizability axiom asserts that the intrinsic topology (see Section \ref{sec:synth-top}) on $\mathbb{R}$ agrees with the usual metric topology.
It contradicts classical logic, and is not satisfied in most models of our logic, i.e.\@ $\Pi W$-pretoposes.
However, it does hold in the big topos of topological spaces and in the K2 realizability topos (see Section \ref{sec:lebesgue}).

\paragraph{Semantics}

As we work internally in a rather unusual logic, we have to justify its consistency, in particular the consistency of the two axioms that we assume on top of predicative mathematics.
First of all, a model of predicative mathematics is a $\Pi W$-pretopos \citep{RijkeSpitters}.
Every topos is, in particular, a $\Pi W$-pretopos, and so our results apply to the internal logic of every topos that validates our axioms.
\cite{lesnikphd} proves that both the K2 realizability topos and the topological toposes \citep{fourman1984continuous,vanderHoevenMoerijk} validate our two assumptions and can thus interpret all of our constructions.

The topos used
in \citet{fourman1984continuous,Fourman2013ContinuousTI} and the topos of continuous
$M$-actions for the localic monoid of endomorphisms of Baire space
used in \citet{vanderHoevenMoerijk} are equivalent by the
Comparison Lemma \citep[Theorem C.2.2.3]{johnstone2002sketches} because the
topological monoid $M$ is dense in the site of separable locales, all of which can be covered by Baire space.
Thus sheaves in the latter topos can be seen as a uni-typed version of sheaves in the former topos.
Both of these works provide a constructive elaboration of Brouwer's
continuity principles.

Our specific realization of predicative mathematics is the theory of $0$-truncated types, \emph{sets}, in HoTT.
This, however, poses new semantical problems:
$\Pi W$-pretoposes and in particular 1-toposes do not model HoTT's higher dimensional priniciples such as the univalence axiom.

While we expect that both the K2 realizability topos and the topological toposes can be suitably embedded into models of HoTT, only the case of topological toposes appears to be resolved (though see e.g.\@ \citet{swan2019church} for progress on realizability models):
It was proved by \citet{shulman2019infty1toposes} that most of HoTT as presented in the HoTT book can be interpreted in all Grothendieck $\infty$-toposes.
Shulman's $\infty$-topos models can also interpret propositional resizing (impredicativity), and so Assumption \ref{ass:free-omega-cpo} holds in these models.
Every Grothendieck $1$-topos of sheaves over a finitely complete site is equivalent to the category of $0$-truncated objects in the corresponding $\infty$-topos \citep{anel-enveloping}.
In particular, this holds for sheaves over the site of small topological spaces, thus the model of HoTT in this $\infty$-topos also validates Assumption \ref{ass:reals-metrizable}.

\paragraph{Order theory}
Let us review some basic notions from order and domain theory.
A \emph{preorder} consists of a carrier set $P$ and a transitive and reflexive relation $x \leq y$ on $P$.
We generally identify a preorder with its carrier set $P$, leaving the order relation implicit.
A map $f : P \rightarrow Q$ of preorders is \emph{monotone} if $x \leq y$ implies $f(x) \leq f(y)$ for all $x, y \in P$.
A \emph{partial order} is a preorder whose ordering relation is antisymmetric.
A \emph{suborder} of a partial order $P$ is a monotone map $i : P' \hookrightarrow P$ with $P'$ a partial order such that $i(x) \leq i(y)$ implies $x \leq y$.
Suborders of $P$ may be identified with subsets of $P$.

Let $I$ and $P$ be preorders and let $d : I \rightarrow P$ be a monotone map.
The \emph{join} $\bigvee d = \bigvee_{i \in I} d(i)$ of $d$ is a least element such that $d(i) \leq \bigvee d$ for all $i \in I$.
Dually, a \emph{meet} $\bigwedge d = \bigwedge_{i \in I} d(i)$ is a greatest element such that $d(i) \geq \bigwedge d$ for all $i \in I$.
Joins and meets are uniquely determined up to isomorphism, i.e.\@ if $e$ and $e'$ are both joins (or both meets) of the same diagram $D$, then $e \leq e'$ and $e' \leq e$.
If we say that certain kinds of joins or meets exist in a preorder, we mean that there is a function that assigns to every suitable diagram its join or meet, respectively, and these canonical joins and meets are denoted by $\bigvee - $ and $\bigwedge -$.
If $P$ is a partial order, then joins and meets are unique if they \emph{merely} exist, and so by unique choice we obtain unique join and meet functions.

Identifying subsets $U \subseteq P$ with suborders of $P$, we write $\bigvee U \in P$ for the join over the corresponding inclusion map.
A monotone map $f : I' \rightarrow I$ of preorders $I'$ and $I$ is \emph{final} if for each $i \in I$ there merely exists $i' \in I'$ such that $i \leq f(i')$.
If $d : I \rightarrow P$ is a monotone map into a partial order $P$ and $f : I' \rightarrow I$ is final, then the two joins $\bigvee d$ and $\bigvee \, (d \circ f)$ exist and agree if either one exists.

A preorder $I$ is \emph{directed} if $I$ is inhabited and there is a function $u : I \times I \rightarrow I$ (not necessarily monotone) such that for all $i, j \in I$ we have $i \leq u(i, j)$ and $j \leq u(i, j)$.
The partial order $\omega$ has for its carrier set the natural numbers with its natural order (which is generated by $n \leq n + 1$ for all $n$).
If $I$ is enumerable (i.e.\@ there exists a surjection $\mathbb{N} \twoheadrightarrow I$) and directed, then there exists a final map $\omega \rightarrow I$.
Thus enumerable directed joins in partial orders $P$ can be reduced to joins over maps $\omega \rightarrow P$, i.e.\@ chains $x_0 \leq x_1 \leq \dots$ in $P$.

\emph{Bottom} and \emph{top} elements are joins $\bot = \bigvee \emptyset$ and meets $\top = \bigwedge \emptyset$, respectively, over the empty set.
A \emph{lattice} is a partial order $L$ which has all binary joins $x \vee y = \bigvee \{x, y\}$ and binary meets $x \wedge y = \bigwedge \{x, y\}$ for $x, y \in L$.
It is \emph{distributive} if $x \wedge (y \vee z) = (x \wedge y) \vee (y \wedge z)$ holds for all $x, y, z \in L$.
An \emph{$\omega$-complete partial order ($\omega$-cpo)} is a partial order which has all enumerable directed joins.
A monotone map $f : C \rightarrow D$ of $\omega$-cpos $C$ and $D$ is \emph{$\omega$-(Scott-)continuous} if $f$ preserves enumerable directed joins.
A \emph{$\sigma$-frame} is a partial order with bottom and top elements, binary meets and enumerable joins which satisfy the distributivity law $x \land \bigvee_{n \in \mathbb{N}} y_n = \bigvee_{n \in \mathbb{N}} \, (x \land y_n)$.
A partial order $P$ is a $\sigma$-frame if and only if it has top and bottom elements and is both a distributive lattice and an $\omega$-cpo:
Arbitrary enumerable joins can be computed as $\bigvee_{n \in \mathbb{N}} x_n = \bigvee_{n \in \omega} \, (x_0 \vee \dots \vee x_n)$ using just the lattice and $\omega$-cpo structure.

Sets of truth values $\Omega = \Omega_i$ are partially ordered by implication.
They are stable under joins (disjunctions) and meets (conjunctions) over small indexing sets.

\section{Presentations of $\omega$-cpos}
\label{sec:presentations}

In this section we adapt the notion of \emph{dcpo presentation} described in \citet{jung2008presenting} for $\omega$-cpo presentations.
We discuss three proofs of the existence of free $\omega$-cpo completions, and construct presentations of product $\omega$-cpos.

\begin{definition}
  \label{def:omega-cpo-presentation}
  An \emph{$\omega$-cpo presentation} consists of a preorder $P$ and a \emph{cover relation} $\triangleleft \subseteq P \times \mathcal{P}(P)$ such that $p \triangleleft U$ ($p$ is \emph{covered} by $U$) holds only if $U$ is an enumerable directed suborder of $P$ (thus $U$ is given by a map $\mathbb{N} \rightarrow P$ with directed image).
  We generally leave the covering relation $\triangleleft$ implicit and refer to the $\omega$-cpo presentation $(P, \triangleleft)$ as just $P$.
  A morphism of \emph{$\omega$-cpo presentations} $f : P \rightarrow Q$ is a monotone map preserving covers, in the sense that if $p \triangleleft U$ holds in $P$, then $f(p) \triangleleft f(U)$ holds in $Q$ for all $p \in P$ and $U \subseteq P$.

  Every $\omega$-cpo $C$ can be regarded as an $\omega$-cpo presentation with cover relation
  \begin{equation}
    c \triangleleft U \iff c \leq \bigvee U
  \end{equation}
  for $U \subseteq C$ directed and enumerable.
  $\omega$-continuous maps $C \rightarrow D$ of $\omega$-cpos may be identified with their morphisms when considered as $\omega$-cpo presentations.
\end{definition}

\begin{assumption}
  \label{ass:free-omega-cpo}
  Let $P$ be an $\omega$-cpo presentation.
  Then there is a free $\omega$-cpo over $P$, i.e.\@ there is a morphism $\eta : P \rightarrow P_\omega$ of $\omega$-cpo presentations with $P_\omega$ an $\omega$-cpo such that for any given morphism $f : P \rightarrow C$ with $C$ an $\omega$-cpo there is a unique $\omega$-continuous map $\bar f : P_\omega \rightarrow C$ such that $\bar f \eta = f : P \rightarrow C$.
\end{assumption}
It appears that Assumption \ref{ass:free-omega-cpo} is independent of constructive predicative mathematics.
However, it follows from rather weak additional mathematical principles, all of which are generally considered constructive.

As a first option, one can work with propositional resizing (impredicativity) \citep{hottbook}, i.e.\@ assume that the inclusions $\Omega_0 \subseteq \Omega_1 \subseteq \dots$ are equalities.
Working impredicatively, \citet{jung2008presenting} construct free dcpos over dcpo presentations.
We sketch a straightforward adaptation of their proof for $\omega$-cpos.
Say a lower subset $\mathfrak{a} \subseteq P$ is an \emph{ideal} if from $p \triangleleft U$ and $U \subseteq \mathfrak{a}$ it follows that $p \in \mathfrak{a}$.
Let $\mathrm{Idl}(P)$ be the partial order of all ideals.
Ideals are closed under arbitrary intersections, so every subset $M \subseteq P$ is contained in the least ideal containing it:
\begin{equation}
  \langle M \rangle = \bigcap \{ \mathfrak{a} \in \mathrm{Idl}(P) \mid M \subseteq \mathfrak{a} \}.
\end{equation}
It follows that $\mathrm{Idl}(P)$ has all joins and that they can be computed as $\bigvee_{i \in I} \mathfrak{a}_i = \langle \bigcup_{i \in I} \mathfrak{a}_i \rangle$.
Assigning to each $p \in P$ the principal ideal $\langle \{ q \in P \mid q \leq p \} \rangle$ gives a monotone map from $P$ to $\mathrm{Idl}(P)$ which preserves covers.
It exhibits $\mathrm{Idl}(P)$ as the free suplattice over $P$, i.e.\@ the free partial order with all joins subject to the cover relations.
Now $P_\omega$ can be defined as the least subset of $\mathrm{Idl}(P)$ which contains the principal ideals that is closed under joins of enumerable directed families.

Next, $P_\omega$ can be constructed as a quotient inductive inductive type (QIIT) \citep{altenkirch2018quotient} in homotopy type theory.
The special case of the free $\omega$-cpo with bottom element over a set (i.e.\@ discrete partial order without covers) is worked out in \citet{Altenkirch}.
Given a set $A$, they define $A_\bot$ and a dependent predicate $\leq : A_\bot \times A_\bot \rightarrow \Omega$ mutually recursive as a QIIT.
Elements of $A_\bot$ and their equalities are generated by the constructors
\begin{mathpar}
  \eta : A \rightarrow A_\bot
  \and \bigvee : \left(\sum_{x : \mathbb{N} \rightarrow A_\bot} \prod_{n : \mathbb{N}} x_n \leq x_{n + 1}\right) \rightarrow A_\bot
  \\ \bot : A_\bot
  \and \alpha : \prod_{x, y : A_\bot} x \leq y \rightarrow y \leq x \rightarrow x = y.
\end{mathpar}
$\leq$ has constructors corresponding to reflexivity, transitivity and the universal properties of $\bot$ and $\bigvee$.
The recursion principle for $A_\bot$ as QIIT is the universal property of the free domain over $A$.
This argument can easily be adapted for our purpose:
To construct $P_\omega$ given an $\omega$-cpo presentation $P$, one omits from the scheme defining $P_\bot$ the constructor $\bot$ and adds constructors $\prod_{p, q : P} p \leq q \rightarrow \eta(p) \leq \eta(q)$ corresponding to monotonicity of $\eta$ and
\begin{equation}
  \prod_{p : P} \prod_{U \in \mathcal{P}(P)} p \triangleleft U \rightarrow \eta(p) \leq \bigvee c_U
\end{equation}
where $c_U : \mathbb{N} \rightarrow P$ is a monotone and final map into $U$.
The semantics of QIITs are not entirely understood, but it is proved in \citet{lumsdaine_shulman} that all Grothendieck $\infty$-topos models validate the existence of many HITs.
Work on reducing QIITs to such simpler inductive constructions is ongoing; see \citep{altenkirch2018quotient}.

As a third alternative, $P_\omega$ can be constructed as a quotient of the set $\mathrm{Hom}(\omega, P)$ of monotone sequences in $P$ if one is willing to assume the axiom of countable choice, at least in the important special case where the covering relation is such that $p \triangleleft U$ holds only if $u \leq p$ for all $u \in U$, which is true in all our applications.
A similar construction for $A_\bot$ is worked out in \citet{Altenkirch}, with the general idea going back to \citet{rosolini:dominance}.
Let $\leq'$ be the preorder on the set of monotone functions $\mathrm{Hom}(\omega, P)$ which is generated from $c \leq' d$ if for all $n$ there merely exists $m$ such that $c_n \leq d_m$, and $\eta(p) \leq' c_U$ whenever $p \triangleleft U$, where $\eta(p)$ denotes the constant sequence with value $p$ and $c_U$ is a final sequence in $U$.
If $c, d: \omega \rightarrow P$ are monotone and $c \leq' d$, then it can be shown by induction over transitivity of $\leq'$ that for all $m, n$ there merely exist either $m'$ or $n'$ such that $c(m')$ respectively $d(n')$ is an upper bound for both $c(m)$ and $d(n)$.
It follows that the image of the set-theoretic transpose $\bar c : \mathbb{N} \times \mathbb{N} \rightarrow P$ of a monotone function $c : \omega \rightarrow \mathrm{Hom}(\omega, P)$ ($\bar c$ need not be monotone with respect to the product order) is directed:
The \emph{mere} existence of binary upper bounds implies the existence of a function assigning upper bounds because of the bijection $\mathbb{N} \times \mathbb{N} \cong \mathbb{N}$ and countable choice.
We obtain a final sequence $c' : \omega \rightarrow P$, which can be shown to be a join of $c$.
Let $P_\omega$ be the quotient partial order of the preorder $(\mathrm{Hom}(\omega, P), \leq')$.
By countable choice, every sequence $c : \omega \rightarrow P_\omega$ can be lifted to one in $\mathrm{Hom}(\omega, P)$, where its join can be computed and mapped back to $P_\omega$.
Thus $P_\omega$ is an $\omega$-cpo, and the verification of its universal property is straightforward.

\begin{proposition}
  \label{prop:enriched-omega-cpo-completion}
  The free $\omega$-cpo completion is monotone on functions:
  If $f \leq g : P \rightarrow Q$, then $f_\omega \leq g_\omega : P_\omega \rightarrow Q_\omega$.
\end{proposition}
\begin{proof}
  The subset $\{ x \in P_\omega \mid f_\omega(x) \leq g_\omega(x) \}$ contains $\eta(p)$ for all $p \in P$ and is closed under directed enumerable joins.
\end{proof}

\citet[proposition 2.8]{jung2008presenting} construct presentations of product dcpos based on presentations of their factors, and an analogous result holds for $\omega$-cpos.
Our proof differs slightly from theirs because we do not assume that $\omega$-completions are constructed as sets of ideals and instead rely solely on the universal property.
\begin{proposition}
  \label{prop:product-presentations}
  Let $P$ and $Q$ be $\omega$-cpo presentations.
  Define a cover relation on the product partial order $P \times Q$ by $(p, q) \triangleleft U \times \{ q \}$ if $p \triangleleft U$ in $P$ and $(p, q) \triangleleft \{ p \} \times V$ if $q \triangleleft V$ in $Q$.
  Then the canonical map $f : (P \times Q)_\omega \rightarrow P_\omega \times Q_\omega$ is an order isomorphism.
\end{proposition}
\begin{proof}
  Let $g_0 : P \rightarrow (Q \rightarrow (P \times Q)_\omega)$ be the function assigning to each $p \in P$ the function $q \mapsto \eta(p, q)$.
  The set of functions $Q \rightarrow (P \times Q)_\omega$ is an $\omega$-cpo with joins computed pointwise.
  If $p \triangleleft U$ and $q \in Q$, then $\bigvee_{u \in U} g_0(u)(q) = \bigvee \eta(U) \times \{ q \} \geq \eta(p, q) = g_0(p, q)$ by definition of the cover relation on $P \times Q$.
  Thus $g_0$ preserves covers and induces an $\omega$-continuous map $g_1 : P_\omega \rightarrow (Q \rightarrow (P \times Q)_\omega)$.
  Let $g_2 : Q \rightarrow (P_\omega \rightarrow (P \times Q)_\omega)$ be its transpose; it is valued in $\omega$-continuous functions.
  Suppose $q \triangleleft V$ and let us prove that for each $x \in P_\omega$ we have
  \begin{equation}
    \label{eq:g2-cover-preservation}
    g_2(q)(x) \leq \bigvee_{v \in V} g_2(v)(x).
  \end{equation}
  If $x = \eta(p)$ for some $p \in P$, then this holds because $(p, q) \triangleleft \{ p \} \times V$ in $P \times Q$.
  If \eqref{eq:g2-cover-preservation} holds for every element $x \in W$ for a directed enumerable family $W \subseteq P_\omega$, then
  \begin{equation}
    g_2(q)(\bigvee W) = \bigvee_{x \in W} g_2(q)(x) \leq \bigvee_{x \in W} \bigvee_{v \in V} g_2(v)(x) = \bigvee_{v \in V} g_2(v)(\bigvee W)
  \end{equation}
  because $g_2(q)$ and $g_2(v)$ for all $v$ commute with joins and joins commute among each other.
  Thus $g_2$ preserves covers and induces an $\omega$-continuous map $g_3 : Q_\omega \rightarrow (P_\omega \rightarrow (P \times Q)_\omega)$.
  Let $g : P_\omega \times Q_\omega \rightarrow (P \times Q)_\omega$ be its transpose.

  $g$ is $\omega$-continuous in each argument.
  Thus if $p : I \rightarrow P_\omega$ and $q : I \rightarrow Q_\omega$ are monotone maps with $I$ enumerable and directed, then
  \begin{equation}
    g(\bigvee_{i \in I} \, (p_i, q_i)) = \bigvee_{i \in I} \bigvee_{j \in I} g(p_i, q_j) = \bigvee_{k \in I} g(p_k, q_k)
  \end{equation}
  because, $I$ being directed, the diagonal $I \rightarrow I \times I$ is final.
  It follows that $g$ is $\omega$-continuous.
  Thus $g f$ is the identity by the universal property of the $\omega$-cpo completion, and $f g = \mathrm{id}$ holds by the universal property of products.
\end{proof}

\begin{corollary}
  Let $P$ be an $\omega$-cpo presentation.
  If $P$ has a bottom element $\bot$, then $\eta(\bot) \in P_\omega$ is a bottom element, and likewise for top elements.
  If $P$ has all binary joins which are compatible with covers in the sense that $\vee : P \times P \rightarrow P$ preserves the covers on $P \times P$ defined in Proposition \ref{prop:product-presentations}, then $P_\omega$ has all binary joins and $\eta : P \rightarrow P_\omega$ preserves them.
  The same is true for binary meets.
\end{corollary}
\begin{proof}
  Without loss of generality, we may assume that for all $p \in P$ we have $p \triangleleft \{ p \}$ because adding these covers to $P$ does not change the generated $\omega$-cpo $P_\omega$.
  Endow the terminal partial order $1$ with the covering relation $* \triangleleft \{ * \}$, where $* \in 1$ is the unique element of the unit set.
  Then the map $P \rightarrow 1$ is a map of $\omega$-cpo presentations, and so are its right or left adjoints $1 \rightarrow P$ if they exist.
  Because $1_\omega = 1$ and the $\omega$-cpo completion is monotone (Proposition \ref{prop:enriched-omega-cpo-completion}), it follows that $P_\omega \rightarrow 1$ is a right (left) adjoint if $P \rightarrow 1$ is.
  Thus $P_\omega$ has a bottom (top) element if $P$ has one.

  Suppose $p \triangleleft U$ in $P$.
  Then
  \begin{equation}
    (\eta(p), \eta(p)) \leq \bigvee_{u \in U} \bigvee_{v \in U} (\eta(u), \eta(v)) = \bigvee_{w \in U} (\eta(w), \eta(w))
  \end{equation}
  because $U$ is directed.
  We may thus add the diagonal covers
  \begin{equation}
    \label{eq:diagonal-cover}
    (p, p) \triangleleft \{(u, u) \mid u \in U\}
  \end{equation}
  to the covers of $P \times P$ without changing the generated $\omega$-cpo.
  Because $P \times P$ presents the product $P_\omega \times P_\omega$, the diagonal $P_\omega \rightarrow P_\omega \times P_\omega$ is obtained by $\omega$-cpo completion of the diagonal of $P$.
  Now suppose $P$ has binary joins which preserve the covers defined in Proposition \ref{prop:product-presentations}.
  Binary joins will always preserve diagonal covers as in \eqref{eq:diagonal-cover}.
  Thus the binary join map can be extended to a left adjoint to the diagonal of $P_\omega$, i.e.\@ $P_\omega$ has binary joins.
  Similarly, if $P$ has a cover preserving binary meet map, then its extension to $P_\omega$ will be right adjoint to the diagonal.
\end{proof}

\section{Synthetic topology and the initial $\sigma$-frame}
\label{sec:synth-top}

In synthetic topology \citep{hyland1991first,Escard2004SyntheticTO,lesnikphd} one works with sets and functions as if they behave like topological spaces and continuous maps.
For this analogy to have any value, the very least one would expect is a notion of \emph{open subset} of a given set (i.e.\@ space).
The set of (small) subsets of a given set $A$ is given by the set of functions $A \rightarrow \Omega$.
It is thus natural to expect a subset $S \subseteq \Omega$ that classifies the \emph{open} subsets, in the sense that a function $A \rightarrow \Omega$ is the indicator function of an open subset if and only if it factors via $S$.
$S$ may be thought of as the \emph{set of open truth values}.
We obtain sets $\mathcal{O}(A) = (A \rightarrow S)$ of open subsets for every set (space) $A$, and it can indeed be verified that the preimage of an open subset under every function is again open.
Thus all functions are continuous.

In traditional (analytic) topology, $S$ corresponds to the \emph{Sierpinski space}:
The space with carrier $\Omega$ whose only nontrivial open is the singleton set $\{ \top \}$.
Indicator functions $\chi : A \rightarrow \Omega$ with $A$ a  topological space (in the usual sense) are continuous if and only if the preimage of $\top$ is open; in other words if and only if $\chi$ corresponds to an open subset.

Without imposing any further requirements on $S$, there is not much we can say about the sets $\mathcal{O}(A)$.
For example, $S = \emptyset$ might be empty, in which case only the empty subset has any open subsets at all.
If $S = \{ \top \}$, then $\mathcal{O}(A) = \{ A \}$ for all $A$.
For $S = \mathbb{B} = \{ \bot, \top \}$ the booleans, the opens are precisely the decidable subsets.
In this case, $S$ is closed under finite conjunctions and disjunction, corresponding to open subsets being closed under finite intersections and unions.
But in constructive models, the booleans are usually not closed under infinite conjunction, so we may not assume that any infinite unions of opens are open.
Arguably the most interesting case is where $S$ is a proper subset of $\Omega$ (so that the topology is not discrete), contains the boolean truth values $\top$ and $\bot$ and is closed under \emph{enumerable} disjunction.
This makes it possible to study limits and first-countable spaces such as the real numbers, which are at the heart of integration theory.
Following the HoTT book and \citet{Gilbert}, we take for $S$ the least subset of $\Omega$ satisfying these constraints:
The initial $\sigma$-frame.
\begin{definition-proposition}[\cite{Gilbert}]
  \label{def-prop:sierpinsky}
  The \emph{Sierpinski space} $\mathbb{S} = \mathbb{B}_\omega$ is the free $\omega$-cpo over the partial order $\mathbb{B} = \{ \bot \leq \top \}$ of decidable truth values.
  $\mathbb{S}$ admits the structure of a $\sigma$-frame, and it is the initial one.
  The map $\mathbb{S} \rightarrow \Omega$ given by $s \mapsto s = \top$ exhibits $\mathbb{S}$ as a suborder of $\Omega$ and preserves all $\sigma$-frame structure.
\end{definition-proposition}
Thus $\mathbb{S}$ is a suborder of $\Omega$, and we freely identify elements $s \in \mathbb{S}$ with their image in $\Omega$.
The preservation of enumerable joins by the inclusion $\mathbb{S} \subseteq \Omega$ means that if $\bigvee_{n \in \mathbb{N}} s_n = \top$ holds for an enumerable family of elements $s_n \in \mathbb{S}$, then there merely exists $n$ such that $s_n = \top$.

As explained in Section \ref{sec:presentations}, in the presence of countable choice $\mathbb{S}$ may be identified with monotone binary sequences $\omega \rightarrow \mathbb{B}$, where we distinguish sequences only by whether they eventually reach $\top$.
This set is also known as the \emph{Rosolini dominance} \citep{rosolini:dominance} and denoted by $\Sigma_1^0$.
When $\mathbb{S} = \Sigma_1^0$, open subsets $U : A \rightarrow \mathbb{S}$ can be understood as the semi-decidable subsets.
Let $a \in A$ and let $s_0 \leq s_1 \leq \dots$ be an increasing binary sequence representing $U(a)$.
If $s_n = \top$ for some $n$, then $a \in U$, but we can never conclude $a \notin U$ by checking only a finite prefix of $s$.
Under a realizability interpretation, $s$ corresponds to a computation producing an infinite stream of digits which will eventually contain $1$ if and only if $a \in U$.
If furthermore $A$ itself is enumerable, we obtain an enumeration of $U$.
The Rosolini dominance is not well-behaved without countable choice.
For example, it is not closed under enumerable disjunction.
We circumvent this issue by using the initial $\sigma$-frame instead, which is closed under enumerable disjunction by definition.

An important requirement imposed on the set of open truth values is the \emph{dominance axiom}.
Consider inclusions of spaces $A \subseteq B \subseteq C$ such that $A$ is open in $B$ and $B$ is open in $C$.
In analytic topology, this implies that $A$ is open in $C$.
This is not automatic in synthetic topology, but holds if $S \subseteq \Omega$ is a \emph{dominance} \citep{rosolini:dominance}:
\begin{definition}[\href{https://github.com/FFaissole/Valuations/blob/d06d2c8c9cce3ddf6137ca3440ab02031912d292/Dominance.v\#L23}{\tt Dominance.v:23}]
  A subset $S \subseteq \Omega$ is a \emph{dominance} if for all $p \in \Omega$ and $s \in S$ it holds that
  \begin{equation}
    \label{eq:dominance}
    (s \implies (p \in S)) \implies (s \land p) \in S.
  \end{equation}
\end{definition}
Note that $p \in S$ and $(s \land p) \in S$ are themselves propositions, hence elements of $\Omega$.
Elements $s \in S$ are, via the inclusion $S \subseteq \Omega$, in particular propositions.

\citet{rosolini:dominance} proved that $\Sigma_1^0$ is a dominance under the assumption of countable choice.
It follows that $\mathbb{S}$ is a dominance if countable choice holds.
But $\mathbb{S}$ being a dominance can be proved directly, and even without assuming countable choice:
\begin{theorem}[\href{https://github.com/FFaissole/Valuations/blob/d06d2c8c9cce3ddf6137ca3440ab02031912d292/Dominance.v\#L32}{\tt Dominance.v:32}]
  \label{th:dominance}
  The Sierpinski space $\mathbb{S} \subseteq \Omega$ is a dominance.
\end{theorem}
\begin{proof}
  We prove the dominance property \eqref{eq:dominance} for fixed $p \in \Omega$ using the induction principle of $\mathbb{S}$ as a free $\omega$-cpo completion of $\mathbb{B}$.
  If $s = \top$ and $s \implies (p \in \mathbb{S})$, then in particular $p \in \mathbb{S}$ and thus $(s \land p) = p$ is in $\mathbb{S}$.
  If $s = \bot$, then $(s \land p) = \bot$, which is an element of $\mathbb{S}$.
  Now let $s = \bigvee_n s_n$ for an ascending chain $s_0 \leq s_1 \leq \dots$ in $\mathbb{S}$.
  Suppose that $s \implies (p \in \mathbb{S})$ and that the dominance property \eqref{eq:dominance} with $s_n$ in place of $s$ holds for all $n \in \mathbb{N}$.
  Combining this with $s_n \implies s$ and $s \implies p$ it follows that $s_n \land p$ is in $\mathbb{S}$ for all $n$.
  Thus
  \begin{equation}
    s \land p = (\bigvee_n s_n) \land p = \bigvee_n (s_n \land p)
  \end{equation}
  by the distributive law, which is in $\mathbb{S}$.
\end{proof}

Given a dominance $S$ and a set $A$, Rosolini constructs a \emph{partial map classifier} of $A$, which is an object representing partial maps $B \rightharpoonup A$ whose domains of definition are open with respect to $S$.
Following \citet{EscardoKnapp}, the partial map classifier can be defined as
\begin{equation}
  \mathcal{L}_S A = \{(s, v) \mid s \in S, v : s \rightarrow A \}.
\end{equation}
Here $s$ is identified with the subsingleton set $\{ * \mid s \}$.
They refer to elements $(s, v) \in \mathcal{L}_S A$ as \emph{partial elements}.
$v$ is the \emph{value}, $s$ its \emph{extent}.
Under a realizability interpretation and $S = \mathbb{S} = \Sigma_1^0$, maps $B \rightarrow \mathcal{L}_S A$ can be thought of as partial functions from $B$ to $A$, in the sense that their interpretations yield potentially non-terminating computations producing results in $A$.
The interpretation of constructive logic in the effective topos even validates the axiom that for every function $\mathbb{N} \rightarrow \mathcal{L}_S \mathbb{N}$ there merely exists a Turing machine which computes it \citep[chapter 3]{bridges1987varieties}.

If one uses the booleans $\{ \bot, \top \}$ as set of open truth values $S$, then $\mathcal{L}_S A$ is the set of \emph{decidably} partial elements.
$\mathcal{L}_S A$ can then be described as the free partial order with bottom element over the discrete partial order $A$.
Its underlying set is the sum $A + 1$, and all elements of $A$ are greater than the element of $1$.
In this case it is thus decidable whether $x \in \mathcal{L}_S A$ represents a fully defined element (i.e.\@ $x \in A$) or whether $x$ is undefined (i.e.\@ $x \in 1$), so that we may think of elements of $A + 1$ as \emph{decidably} partial elements of $A$.

\citet{Altenkirch} propose defining the partial map classifier of $A$ as the QIIT $A_\bot$ described in Section \ref{sec:presentations}.
In our terminology, $A_\bot$ is the $\omega$-cpo completion $(A + 1)_\omega$.
\citet{EscardoKnapp} mention that $A_\bot$ can be understood in terms of Rosolini's lifting construction.
Indeed, $\mathcal{L}_\mathbb{S} A$ has the structure of an $\omega$-cpo with bottom element under $A$:
The structure map $e : A \rightarrow \mathcal{L}_\mathbb{S} A$ is defined by assigning to each element $a \in A$ the unique map $\top \rightarrow A$ with value $a$.
For $v : s \rightarrow A$ and $v' : s' \rightarrow A$ in $\mathcal{L}_\mathbb{S} A$ let
\begin{equation}
  (s, v) \leq (s', v') \iff ((s \implies s') \land v'_{| s} = v : s \rightarrow A.
\end{equation}
This defines a partial order on $\mathcal{L}_\mathbb{S} A$.
Its bottom element is the unique map $\bot \rightarrow A$. 
The join of an enumerable directed set $U = \{ (s_u, v_u) \mid u \in U \} \subseteq \mathcal{L}_\mathbb{S} A$ is given by $(\bigvee_{u \in U} s_u, v)$, where $v$ is defined by $v(x) = v_{u_0}(x)$ whenever $x \in \bigvee_{u \in U} s_u$ is in $s_{u_0}$.
Thus there is a unique $\omega$-continuous map $f : A_\bot \rightarrow \mathcal{L}_\mathbb{S} A$ which is compatible with the structure maps and preserves the bottom element.
We can then show the following:
\begin{proposition}
  The map $f : A_\bot \rightarrow \mathcal{L}_\mathbb{S} A$ is an order isomorphism.
\end{proposition}
\begin{proof}
  First note that the projection $\mathcal{L}_\mathbb{S} A \rightarrow \mathbb{S}$ that sends a partial element $(s, v)$ to its extent $s$ is $\omega$-continuous and preserves the bottom element.
  The unique map $A \rightarrow 1$ induces a map $A_\bot \rightarrow 1_\bot = \mathbb{S}$, which can equivalently be described as assigning to $x \in A_\bot$ the truth value
  \begin{equation}
    (\exists a \in A, \eta(a) = x) \in \Omega
  \end{equation}
  by Proposition \ref{def-prop:sierpinsky}.
  (A direct proof of this can also be found in \citet{Gilbert}.)
  Here the existential quantifier denotes \emph{mere} existential quantification.
  By the universal property of $A_\bot$, the maps constructed so far commute with $f$, so if $f(x) = (s, v)$, then $s \iff \exists a \in A, x = \eta(a)$.

  Now let us show that $f$ exhibits $A_\bot$ as suborder of $\mathcal{L}_\mathbb{S} A$.
  Suppose $f(x) = (s, v)$ and $f(x') = (s', v')$ such that $(s, v) \leq (s', v')$ in $\mathcal{L}_\mathbb{S} A$.
  We show $x \leq x'$ by induction over $x$.
  If $x = \bot$, then trivially $x \leq x'$.
  If $x = \eta(a)$ for some $a \in A$, then $s' \geq s = \top$, hence $s' = \top$.
  From this it follows by our initial remark that there merely exists $a' \in A$ such that $x' = \eta(a')$.
  In particular, $a = v(*) = v'(*) = a'$, where $* \in \top$ is the unique element of the unit set, hence $x = x'$.
  Now let $x = \bigvee U$ be the join of a directed enumerable subset $U \subseteq A_\bot$.
  We may assume that for all $u \in U$, if $f(u) \leq f(x')$, then $u \leq x'$.
  Thus $u \leq x'$ because $f(u) \leq f(x) \leq f(x')$ for all $u$.
  But then $x = \bigvee U \leq x'$ by definition of least upper bound.

  It remains to show that $f$ is surjective and hence an order isomorphism.
  For this we must construct for each partial element $(s, v) \in \mathcal{L}_\mathbb{S} A$ an element $x \in A_\bot$ such that $f(x) = (s, v)$.
  We proceed by induction over $s$.
  We can set $x = \bot$ if $s = \bot$ and $x = \eta(v(*))$ if $s = \top$.
  Now let $s = \bigvee U$ be a directed enumerable join in $\mathcal{L}_\mathbb{S} A$.
  We may assume that for partial elements of the form $w : u \rightarrow A$ with $u \in U$ there merely exists $x \in A_\bot$ such that $f(x) = (u, w)$.
  Because $f : A_\bot \rightarrow \mathcal{L}_\mathbb{S} A$ was already proved to be the inclusion of a suborder,
  \begin{equation}
    V = \{ x \in A_\bot \mid f(x) = (u, v_{| u}) \text{ for some } u \in U \}
  \end{equation}
  embeds into $U$.
  By the induction hypothesis, it is isomorphic to $U$, hence directed and enumerable.
  Now $f(\bigvee V) = \bigvee f(V) = \bigvee_{u \in U} (u, v_{| u}) = (s, v)$.
\end{proof}

\section{The lower reals}
\label{sec:lower-reals}

A \emph{Dedekind cut} is a pair of sets of rational numbers $(L, U)$ of the form $L = (\infty, x) \cap \mathbb{Q}$ and $U = (x, \infty) \cap \mathbb{Q}$ for some real number $x$.
The condition that $(L, U)$ is of this form can be stated purely in terms of rational numbers without referring to the real numbers, so the (Dedekind) real numbers $\mathbb{R}$ can be defined as the set of all pairs $(L, U)$ satisfying these requirements; see e.g.\@ \citet{johnstone2002sketches}.
Constructively, even a bounded subset of $\mathbb{R}$ does not necessarily have a supremum.
This is problematic in integration theory, because integrals of functions on non-compact spaces are constructed by approximating them from below.

A lower real is given only by the lower part $L$.
Note that, constructively, $U$ cannot be reconstructed from just $L$ or vice-versa.
In the setting of synthetic topology, it is natural to ask that the subsets $L$ (and $U$) are valued in the Sierpinski space $S$, so that they correspond to subsets of $\mathbb{Q}$ which are open with respect to $S$.
For Dedekind reals, this has been studied extensively by \citet{lesnikphd}.
The usage of the initial $\sigma$-frame $\mathbb{S}$ in the definition of Dedekind real numbers is also proposed in the HoTT book (Section 11.2) and has been formalized by \citet{Gilbert}.
For us $\mathbb{S} = S$ is the Sierpinski space, so real numbers $x$ given by open Dedekind cuts can be understood as those for which the predicates $q < x$ and $q > x$ on rational numbers $q$ are semi-decidable.
If $x$ is a lower real, then only the predicate $q < x$ will be semi-decidable.
We use the symbol $\mathbb{R}$ to refer to the Dedekind reals valued in $\mathbb{S}$ and likewise $\mathbb{R}_l$ for the set of lower reals valued in $\mathbb{S}$.

\begin{definition}[\href{https://github.com/FFaissole/Valuations/blob/d06d2c8c9cce3ddf6137ca3440ab02031912d292/Rlow.v\#L46}{\tt Rlow.v:46}]
  A \emph{lower real} is an open subset $L : \mathbb{Q} \rightarrow \mathbb{S}$ of $\mathbb{Q}$ satisfying the following axioms:
  \begin{itemize}
    \item
      There merely exists $q \in \mathbb{Q}$ such that $L(q)$,
    \item
      for all $q \in \mathbb{Q}$, if $L(q)$ then there merely exists $q' > q$ such that $L(q')$, and
    \item
      for all $q < q' \in \mathbb{Q}$, if $L(q')$, then $L(q)$.
  \end{itemize}
  The set of all lower reals is denoted by $\mathbb{R}_l$.
  For $q \in \mathbb{Q}$ let
  \begin{equation}
    \underline q = \{ p \in \mathbb{Q} \mid p < q \} \in \mathbb{R}_l.
  \end{equation}
  The subset of \emph{non-negative lower reals} is given by
  \begin{equation}
    \mathbb{R}_l^+ = \{L \in \mathbb{R}_l \mid \forall q \in \mathbb{Q}, q < 0 \implies L(q) \}.
  \end{equation}
\end{definition}
Note that $\infty = \mathbb{Q} \in \mathbb{R}_l$ is a lower real, but that $-\infty$ (however it may be defined) is not in $\mathbb{R}_l$.
In predicative foundations, the Dedekind or lower reals usually have to be parameterized by a universe level $i$, corresponding to the size of the set of truth values $\Omega_i$ the lower (and upper) cuts are valued in.
The resulting set of reals will only be an element of the $(i + 1)$th universe.
Using the set of open truth values $\mathbb{S}$, we avoid this nuisance and obtain just one set of Dedekind and lower reals, respectively.

Crucial for the use of lower reals in integration theory is their order-theoretic structure:
\begin{proposition}[\href{https://github.com/FFaissole/Valuations/blob/d06d2c8c9cce3ddf6137ca3440ab02031912d292/Rlow.v}{\tt Rlow.v}]
  \label{prop:lower-real-sigma-frame}
  The lower reals endowed with the relation
  \begin{equation}
    L_1 \leq L_2 \iff \forall q \in \mathbb{Q}, q \in L_1 \implies q \in L_2
  \end{equation}
  for $L_1, L_2 \in \mathbb{R}_l$ are a partial order.
  Finite meets and enumerable joins in $\mathbb{R}_l$ exist, are computed pointwise and satisfy the distributivity law $x \land (\bigvee_{n \in \mathbb{N}} y_n) = \bigvee_{n \in \mathbb{N}} \, (x \land y)$.
  The suborder of non-negative lower reals $\mathbb{R}^+_l$ is a $\sigma$-frame.
  The map $q \mapsto \underline{q}$ exhibits $\mathbb{Q}$ as suborder of $\mathbb{R}_l$.
  \qed
\end{proposition}

In view of Proposition \ref{prop:lower-real-sigma-frame}, it is natural to wonder whether $\mathbb{R}_l$ is obtained by a completion process of $\mathbb{Q}$.
This is indeed the case.
Define a cover relation on $\mathbb{Q}$ by $q \triangleleft U$ for enumerable directed $U \subseteq \mathbb{Q}$ such that $\bigvee U$ exists and is equal to $q$.
The embedding $\mathbb{Q} \subseteq \mathbb{R}_l$ preserves enumerable joins and thus induces an $\omega$-continuous map $f : \mathbb{Q}_\omega \rightarrow \mathbb{R}_l$.
Similarly we have $f^+ : (\mathbb{Q}^+)_\omega \rightarrow \mathbb{R}_l^+$, where $\mathbb{Q}^+$ is understood as an $\omega$-cpo presentation with the restricted cover relation of $\mathbb{Q}$.
\begin{theorem}
  \label{th:lower-reals-as-completion}
  The unique $\omega$-continuous maps $f : \mathbb{Q}_\omega \rightarrow \mathbb{R}_l$ and $f^+ : (\mathbb{Q}^+)_\omega \rightarrow \mathbb{R}^+_l$ under $\mathbb{Q}$ respectively $\mathbb{Q}^+$ are order isomorphisms.
\end{theorem}
Noting that the two operations preserve covers, we conclude with Proposition \ref{prop:product-presentations} the following:
\begin{corollary}[\href{https://github.com/FFaissole/Valuations/blob/d06d2c8c9cce3ddf6137ca3440ab02031912d292/Rlow.v}{\tt Rlow.v}]
  \label{cor:lower-real-arithmetic}
  Addition on $\mathbb{Q}$ and multiplication on $\mathbb{Q}^+$ extend uniquely to $\omega$-continuous operations on $\mathbb{R}_l$ and $\mathbb{R}_l^+$, respectively.
\end{corollary}
Our Coq formalization includes definitions of addition on $\mathbb{R}_l$ (\href{https://github.com/FFaissole/Valuations/blob/d06d2c8c9cce3ddf6137ca3440ab02031912d292/Rlow.v#L330}{\tt Rlow.v:330}) and multiplication of lower reals by rational numbers (\href{https://github.com/FFaissole/Valuations/blob/d06d2c8c9cce3ddf6137ca3440ab02031912d292/Rlow.v#L1376}{\tt Rlow.v:1376}), but does not prove uniqueness as asserted by Corollary \ref{cor:lower-real-arithmetic}.

Multiplication cannot be (constructively) extended to an operation on all lower reals because it is not monotone.
In terms of lower cuts, we have $q \in (L_1 + L_2)$ if and only if there merely exist $q_1 \in L_1$ and $q_2 \in L_2$ such that $q_1 + q_2 = q$, and similarly for the product $L_1 \cdot L_2$ if $L_1, L_2 \in \mathbb{R}_l^+$.

The statement analogous to Theorem \ref{th:lower-reals-as-completion} for the usual lower reals (which are not required to be valued in $\mathbb{S}$) and completion under arbitrary directed joins can be shown as follows.
The proposed inverse $g$ to $f$ maps a lower real $L : \mathbb{Q} \rightarrow \Omega$ to the join $g(L) = \bigvee_{q \in L} \eta(q)$ in the completion of $\mathbb{Q}$ under arbitrary directed joins.
This defines a continuous map which is compatible with the inclusions of $\mathbb{Q}$, hence $gf = \mathrm{id}$ by the universal property of the completion.
On the other hand, $fg = \mathrm{id}$ because $L = \bigvee_{q \in \mathbb{Q}} \underline{q}$ for all $L$.
Unfortunately, this proof does not directly transfer to our situation because lower reals $L : \mathbb{Q} \rightarrow \mathbb{S}$ are not necessarily enumerable in the sense that there is a surjection $\mathbb{N} \twoheadrightarrow L = \{q \in \mathbb{Q} \mid L(q) \}$, at least not in the absence of countable choice.
\begin{proof}[Proof of Theorem \ref{th:lower-reals-as-completion}.]
  For brevity, we only prove the statement about $\mathbb{R}_l$, the proof for $\mathbb{R}_l^+$ being similar.
  Note that the covers of $\mathbb{Q}$ are stable under binary joins, thus $\mathbb{Q}_\omega$ has binary joins and hence arbitrary enumerable joins.
  This allows us to construct a map $g : \mathbb{R}_l \rightarrow \mathbb{Q}_\omega$ as follows.
  Let $L \in \mathbb{R}_l$ and pick $q \in L$.
  For each $p \in \mathbb{Q}$, let $s \mapsto p_s$ be the unique $\omega$-continuous map $\mathbb{S} \rightarrow \mathbb{Q}_\omega$ which sends $\bot$ to $\eta(q)$ and $\top$ to $\eta(p)$.
  Now set
  \begin{equation}
      g(L) = \bigvee_{p \in \mathbb{Q}} p_{L(p)}.
  \end{equation}
  If $p \in L$, then $p_{L(p)} = \eta(p)$ by definition, and so $\bigvee_{q \in \mathbb{Q}} q_{L(q)} \geq \eta(p)$.
  Thus $g$ is well-defined as it does not depend on the choice of $q$.

  $g$ is defined as composition of $\omega$-continuous maps, so is $\omega$-continuous itself.
  It is compatible with the structure maps $\mathbb{Q} \rightarrow \mathbb{R}_l$ and $\mathbb{Q} \rightarrow \mathbb{Q}_\omega$ because
  \begin{equation}
    g(\underline{q}) = \bigvee_{p \in \mathbb{Q}} p_{\underline{q}(p)} = \bigvee_{p < q} \eta(p) = \eta(q)
  \end{equation}
  by definition of the cover relation on $\mathbb{Q}$.
  It follows that $gf = \mathrm{id}$ by the universal property of $\mathbb{Q}_\omega$.

  Note that $f$ preserves arbitrary enumerable joins (not necessarily directed) because the map $\mathbb{Q} \rightarrow \mathbb{R}_l$ preserves binary joins.
  Let $L \in \mathbb{R}_l$.
  It can be shown by induction over $L(p)$ that $f(p_{L(p)}) \leq L$ for all $p \in \mathbb{Q}$.
  Thus
  \begin{equation}
    f(g(L)) = f(\bigvee_{p \in \mathbb{Q}} p_{L(p)}) = \bigvee_{p \in \mathbb{Q}} f(p_{L(p)}) \leq L.
  \end{equation}
  On the other hand, suppose $q \in L$ and let us show that $q \in f(g(L))$, i.e.\@ that $L \leq f(g(L))$.
  Because $L$ is a rounded lower subset of $\mathbb{Q}$, there merely exists $q' > q$ such that $q' \in L$.
  Then $f(q'_{L(q')}) = \underline{q'} \leq f(g(L))$, hence $q \in f(g(L))$.
\end{proof}

\section{Integrals and Valuations}
\label{sec:valuations-integrals}

In this section we define valuations, which play the role of measures but are defined only on opens, and integrals.
We then prove a version of the Riesz theorem, which states that there is a one-to-one correspondence between valuations and integrals.
Valuations are often preferred over measures in constructive mathematics because measures would have to be valued in the hyperreals \citep{coquand2002metric}.
They have a long tradition in the domain-theoretic semantics of probabilistic computations, see e.g.\@ \citet{Jones1989APP}.
It is observed there that, classically, valuations on compact Hausdorff spaces are in bijective correspondence with regular measures.
Our proof of the Riesz theorem is inspired by \citet{integrals-valuations} and \citet{Vickers}, who prove similar results in the setting of locales.

Fix a set $A$.
Recall that $\mathcal{O}(A)$, the set of open subsets of $A$, is defined as the set of functions $A \rightarrow \mathbb{S}$.
The $\sigma$-frame structures of $\mathbb{S}$ and $\mathbb{R}_l^+$ induce $\sigma$-frame structures on the sets of functions $\mathcal{O}(A)$ and $A \rightarrow \mathbb{R}_l^+$, with all structure defined pointwise.

\begin{definition}[\href{https://github.com/FFaissole/Valuations/blob/d06d2c8c9cce3ddf6137ca3440ab02031912d292/Valuations.v\#L49}{\tt Valuations.v:49}]
  An ($\omega$-continuous) \emph{valuation} on a set $A$ is an $\omega$-continuous map $\mu : \mathcal{O}(A) \rightarrow \mathbb{R}_l^+$ preserving the bottom element that satisfies the \emph{modularity law}
  \begin{equation}
    \mu(U) + \mu(V) = \mu(U \cup V) + \mu(U \cap V).
  \end{equation}
  for all opens $U, V \in \mathcal{O}(A)$.
  $\mu$ is a \emph{sub-probability} valuation if $\mu(A) \leq 1$.
  The set of all valuations on $A$ is denoted by $\mathfrak{V}(A)$ and the set of sub-probability valuations by $\mathfrak{V}_{\leq 1}(A)$.
\end{definition}

Let $r : \mathbb{S} \rightarrow \mathbb{R}_l$ be the unique $\omega$-continuous map such that $r(\bot) = 0$ and $r(\top) = 1$.
By postcomposition we obtain a map $(A \rightarrow \mathbb{S}) \rightarrow (A \rightarrow \mathbb{R}_l)$ that assigns to each open $U \in \mathcal{O}(A) = (A \rightarrow \mathbb{S})$ its \emph{(real) indicator function} $\mathds{1}_U = r(U) : A \rightarrow \mathbb{R}_l$.

The map $U \mapsto \mathds{1}_U$ is an order embedding, and so we can equivalently think of a valuation $\mu$ as assigning lower reals to a certain subset of functions $A \rightarrow \mathbb{R}_l^+$.
The Riesz theorem states that every valuation $\mu$ can be extended to a lower integral, which is a function defined on \emph{all} maps $A \rightarrow \mathbb{R}_l^+$, and that every lower integral is determined by its restriction to indicator functions.

\begin{definition}[\href{https://github.com/FFaissole/Valuations/blob/d06d2c8c9cce3ddf6137ca3440ab02031912d292/LowerIntegrals.v\#L76}{\tt LowerIntegrals.v:76}]
  \label{def:lower-integral}
  A \emph{lower integral} on $A$ is an $\omega$-continuous map $\mathcal{I} : (A \rightarrow \mathbb{R}_l^+) \rightarrow \mathbb{R}_l^+$ preserving the bottom element that is furthermore additive, i.e.\@ satisfies
  \begin{equation}
    \mathcal{I}(f + g) = \mathcal{I}(f) + \mathcal{I}(g)
  \end{equation}
  for all $f, g : A \rightarrow \mathbb{R}_l^+$.
  $\mathcal{I}$ is a \emph{sub-probability} lower integral if $\mathcal{I}(\mathds{1}_A) \leq 1$.
  The set of all lower integrals on $A$ is denoted by $\mathfrak{G}(A)$ and the set of sub-probability lower integrals by $\mathfrak{G}_{\leq 1}(A)$.
\end{definition}
The reader might wonder at this point why we need the generality of \emph{sub}-probability valuations and integrals, as opposed to probability valuations and integrals, which would assign to (the indicator function of) the whole space the value 1.
Valuations and integrals on some set $A$ form partial orders, with ordering defined pointwise.
Now, if we restrict to proper probability valuations and integrals, these orders will usually not have least elements (consider, for example, valuations on the set of two elements).
On the other hand, for their sub-probabilistic versions we have the following, which will be crucial for the interpretation of fixpoint operators in Section \ref{sec:interpretation}:
\begin{proposition}
  The inclusions $\mathfrak{V}_{\leq 1}(A) \subseteq \mathfrak{V}(A) \subseteq (\mathcal{O}(A) \rightarrow \mathbb{R}_l^+)$ and $\mathfrak{G}_{\leq 1}(A) \subseteq \mathfrak{G}(A) \subseteq ((A \rightarrow \mathbb{R}_l^+) \rightarrow \mathbb{R}_l^+)$ are embeddings of $\omega$-cpos with bottom elements.
  \qed
\end{proposition}
\begin{proposition}
  \label{prop:integrals-linear}
  Every lower integral $\mathcal{I}$ is compatible with multiplication by scalars from $\mathbb{R}_l^+$, in the sense that $\mathcal{I}(a f) = a \mathcal{I}(f)$ for all $a \in \mathbb{R}_l^+$ and $f : A \rightarrow \mathbb{R}_l^+$.
  In particular, lower integrals are linear over $\mathbb{R}_l^+$.
\end{proposition}
\begin{proof}
  If $a \in \mathbb{N}$, then $\mathcal{I}(a f) = \mathcal{I}(f + \dots + f) = a \mathcal{I}(f)$ because $\mathcal{I}$ is additive.
  Thus if $a = \frac{m}{n}$ is a positive rational, then $n \mathcal{I}(a f) = \mathcal{I}(n a f) = m \mathcal{I}(f)$, hence $\mathcal{I}(a f) = \frac{m}{n}\mathcal{I}(f)$.
  If $U$ is a directed enumerable set of lower reals such that for each $a \in U$ we have $\mathcal{I}(a f) = a \mathcal{I}(f)$ for all $f$, then
  \begin{equation}
    \mathcal{I}((\bigvee U) f) = \mathcal{I}(\bigvee_{a \in U} \, (a f)) = (\bigvee U) \mathcal{I}(f)
  \end{equation}
  by $\omega$-continuity of $\mathcal{I}$ and multiplication, so $\mathcal{I}$ is compatible with multiplication by $\bigvee U$.
  Because $\mathbb{R}_l^+$ is the $\omega$-cpo completion of $\mathbb{Q}^+$ (Theorem \ref{th:lower-reals-as-completion}), it follows that $\mathcal{I}$ is compatible with scalar multiplication by arbitrary non-negative lower reals $a$.
\end{proof}

We are now ready to state the central result of this section.
\begin{theorem}[Riesz]
  \label{th:riesz}
  The assignment
  \begin{equation}
    \mathcal{I} \mapsto (U \mapsto \mathcal{I}(\mathds{1}_U))
  \end{equation}
  defines a map $\mathfrak{G}(A) \rightarrow \mathfrak{V}(A)$ that restricts to a map $\mathfrak{G}_{\leq 1}(A) \rightarrow \mathfrak{V}_{\leq 1}(A)$.
  Both maps are order isomorphisms.
\end{theorem}

We begin the proof by showing that restrictions of lower integrals to indicator functions are valuations.
\begin{lemma}[\href{https://github.com/FFaissole/Valuations/blob/d06d2c8c9cce3ddf6137ca3440ab02031912d292/Riesz1.v\#L22}{\tt Riesz1.v:22}]
  \label{lem:restriction-of-integral-is-valuation}
  Let $\mathcal{I}$ be an integral on $A$.
  Then $\mu_\mathcal{I} : U \mapsto \mathcal{I}(\mathds{1}_U)$ is a valuation on $A$.
  If $\mathcal{I}$ is a sub-probability integral, then $\mu_\mathcal{I}$ is a sub-probability valuation.
\end{lemma}
\begin{proof}
  Recall that $\mathds{1}_U$ is obtained by postcomposing $U : A \rightarrow \mathbb{S}$ with the unique $\omega$-continuous map $r : \mathbb{S} \rightarrow \mathbb{R}_l^+$ that satisfies $r(\bot) = 0$ and $r(\top) = 1$.
  Thus $U \mapsto \mathds{1}_U$ is $\omega$-continuous, too, hence $\omega$-continuity of $\mu_\mathcal{I}$ follows from $\omega$-continuity of $\mathcal{I}$.
  By definition $\mu_\mathcal{I}(A) = \mathcal{I}(\mathds{1}_A)$, so if the latter is $\leq 1$, then so is the former.

  What remains to be shown is that $\mu_\mathcal{I}$ satisfies the modularity law, i.e.\@ that
  \begin{equation}
    \mathcal{I}(\mathds{1}_{U \cup V}) + \mathcal{I}(\mathds{1}_{U \cap V}) = \mathcal{I}(\mathds{1}_U) + \mathcal{I}(\mathds{1}_V).
  \end{equation}
  holds for all $U, V \in \mathcal{O}(A)$.
  By linearity of $\mathcal{I}$ and the definition of indicator functions, it will suffice to show that for all $s, t \in \mathbb{S}$ it holds that
  \begin{equation}
    \label{eq:s-r-modularity}
    r(s \lor t) + r(s \land t) = r(s) + r(t),
  \end{equation}
  and we will do so by induction over $s$.
  If $s = \top$, both sides are equal to $1 + r(t)$, and if $s = \bot$, then both sides are equal to $r(t)$.
  Now let $s = \bigvee U$ for an enumerable directed subset $U \subseteq \mathbb{S}$, and suppose that equation \eqref{eq:s-r-modularity} holds with $u$ in place of $s$ for all $u \in U$.
  Using the fact that the involved operations binary meet and join with $t$, addition and $r$ are all $\omega$-continuous, we compute
  \begin{align}
    r(s \lor t) + r(s \land t)
    ={}& \bigvee_{u \in U} \,(r(u \lor t) + r(u \land t)) \\
    ={} & \bigvee_{u \in U} \, (r(u) + r(t)) \\
    ={} & r(s) + r(t).
  \end{align}
\end{proof}

Next we construct the extension $\int - \, d\mu$ of a valuation $\mu$ to an integral.
Fix $\mu$.
\begin{definition}
  Let $f : A \rightarrow \mathbb{R}_l^+$.
  The lower \emph{$\mu$-integral} $\int f \, d\mu \in \mathbb{R}_l^+$ is defined as follows.
  For $q \in \mathbb{Q}_+$ let
  \begin{equation}
    [f > q] = \{ x \in A \mid q < f(x) \};
  \end{equation}
  it is an open subset of $A$.
  Let
  \begin{equation}
    s_{f, m, n} = \sum_{i = 1}^{mn} \frac{1}{m} \mu([f > \frac{i}{m}])
  \end{equation}
  for $m, n \in \mathbb{N}$.
  Now
  \begin{equation}
    \int f d\mu = \bigvee_{m, n \in \mathbb{N}} s_{f, m, n}.
  \end{equation}
\end{definition}

The main difficulty in showing that $f \mapsto \int f \, d\mu$ is indeed a lower integral is the verification of linearity.
Our main tool will be the \emph{generalized modularity lemma}, originally due to \citet[corollary 1.3]{horn1948measures} in the special case of boolean algebras.
More recent references are \citet{integrals-valuations} and \citet{Vickers}; the latter also contains a proof of the version that will be used here.
Generalized modularity is phrased in terms of the following construction, which in the special case $L = \mathcal{O}(A)$ can be understood as the submonoid of functions $A \rightarrow \mathbb{R}_l^+$ generated by the indicator functions $\mathds{1}_U$ for $U \in \mathcal{O}(A)$.
\begin{definition}
  Let $L$ be a distributive lattice with bottom element.
  The \emph{modular monoid} $M(L)$ is the commutative monoid (written additively) generated by the carrier of $L$ subject to
  \begin{equation}
    a + b = (a \wedge b) + (a \vee b)
  \end{equation}
  for all $a, b \in L$, and $0 = \bot$.
\end{definition}
Note that the modularity law and the preservation of bottom elements guarantee precisely that valuations $\mu : \mathcal{O}(A) \rightarrow \mathbb{R}_l$ factor uniquely as monoid homomorphism $L(\mathcal{O}(A)) \rightarrow \mathbb{R}_l$.

\begin{lemma}[Generalized Modularity Lemma]
  \label{lem:generalized-modularity}
  Let $L$ be a distributive lattice and $x_1, \dots, x_n \in L$.
  Then in $M(L)$ we have
  \begin{equation}
    \sum_{i = 1}^n x_i = \sum_{k = 1}^n \bigvee \{x_I \mid I \subseteq \{1, \dots, n \}, |I| = k\}
  \end{equation}
  where $x_I = \bigwedge \{x_i \mid i \in I\}$ for decidable $I \subseteq \{1, \dots, n\}$.
\end{lemma}

Let $q \in \mathbb{Q}^+$ and $f : A \rightarrow \mathbb{R}_l^+$.
Define $[f > q]_0 \subseteq A$ to be $[f > q]$ if $q > 0$ and equal to $A$ if $q = 0$.
\begin{lemma}
  \label{lem:preimages-of-sums}
  Let $f, g : A \rightarrow \mathbb{R}_l^+$.
  Then in $M(\mathcal{O}(A))$ we have
  \begin{align}
    &\sum_{k = 1}^n ([f > k] + [g > k]) \\
    ={}& \sum_{k = 1}^{2n} \bigvee \{[f > i]_0 \wedge [g > j]_0 \mid i, j \in \{0, \dots, n\}, i + j = k\}
  \end{align}
  for all $n > 0$.
\end{lemma}
\begin{proof}
  Regarding the left-hand side as a sum with $2n$ summands, we have by the Generalized Modularity Lemma \ref{lem:generalized-modularity}
  \begin{align}
    \label{eq:generalized-modularity-application}
    & \sum_{k = 1}^n ([f > k] + [g > k]) \\
    ={} & \sum_{k = 1}^{2n} \bigvee \{[f > I] \wedge [g > J] \mid I, J \subseteq \{1, \dots, n\}, |I| + |J| = k \}
  \end{align}
  where $[f > I] = \bigwedge \{ [f > i] \mid i \in I\}$ and similarly $[g > J] = \bigwedge \{ [g > j] \mid j \in J \}$.
  Because $[f > i_0] \supseteq [f > i_1]$ whenever $i_0 \leq i_1$, we have $[f > I] = [f > \bigvee I]$ for inhabited $I \subseteq \{1, \dots, n\}$.
  If $I$ is empty, then $\bigvee I = 0$ and hence $[f > I] = A = [f > \bigvee I]_0$.
  It follows that $[f > I] \leq [f > \{1, \dots, \ell\}]_0 = [f > \ell]_0$ if $I \subseteq \{1, \dots, n\}$ has $\ell$ elements and similarly for $[g > J]$.
  Discarding small elements from joins, we obtain
  \begin{align}
    & \bigvee \{[f > I] \wedge [g > J] \mid I, J \subseteq \{1, \dots, n\}, |I| + |J| = k \} \\
    ={} & \bigvee \{ [f > i]_0 \wedge [g > j]_0 \mid i, j \in \{0, \dots, n\}, i + j = k \}
  \end{align}
  for $1 \leq k \leq 2n$. 
\end{proof}

\begin{lemma}
  \label{lem:sfmn-directed}
  Let $f : A \rightarrow \mathbb{R}_l^+$.
  Suppose $m, m', n, n'$ are positive integers such that $n \leq n'$ and $m \mathbin{|} m'$ (i.e.\@ $m$ divides $m'$).
  Then $s_{f, m, n} \leq s_{f, m', n'}$.
  The family $(s_{f, m, n})_{m, n}$ is directed.
\end{lemma}
\begin{proof}
  The inequality is clear if $m = m'$, so by transitivity it will suffice to prove the inequality for $n' = n$ and $m' = m q$ for some integer $q > 0$.
  Dividing $i$ by $q$ with remainder, we obtain for each integer $1 \leq i \leq m'n$ unique integers $0 \leq k \leq mn - 1$ and $1 \leq j \leq q$ such that $i = qk + j$.
  Thus
  \begin{equation}
    s_{f, m', n} = \frac{1}{m'} \sum_{i = 1}^{m'n} \mu([f > \frac{i}{m'}]) = \frac{1}{m} \sum_{k = 0}^{mn - 1} \frac{1}{q} \sum_{j = 1}^q \mu([f > \frac{qk + j}{m'}]).
  \end{equation}
  Now $[f > \frac{qk + j}{m'}] \geq [f > \frac{q (k + 1)}{m'}] = [f > \frac{k + 1}{m}]$, hence
  \begin{equation}
    \frac{1}{m} \sum_{k = 0}^{mn - 1} \frac{1}{q} \sum_{j = 1}^q \mu([f > \frac{qk + j}{m'}]) \geq \frac{1}{m} \sum_{k = 0}^{mn - 1} \mu([f > \frac{k + 1}{m}]) = s_{f, m, n}
  \end{equation}
  by monotonicity of $\mu$.
  Both $m \mathbin{|} m'$ and $n \leq n'$ are directed partial orders on the positive integers, thus $(s_{f, m, n})_{m, n}$ is a directed family.
\end{proof}

\begin{lemma}
  \label{lem:mu-int-is-integral}
  Let $\mu$ be a valuation on $A$.
  Then the assignment $f \mapsto \int f \, d\mu$ is a lower integral.
\end{lemma}
\begin{proof}
  We verify the conditions of Definition \ref{def:lower-integral}.

  \emph{Preservation of $\bot$.}
  If $f = 0$ is the constant function with value zero, then $[f > q] = \emptyset$ for all $q > 0$, hence $s_{f,m,n} = 0$ for all $m, n$, so $\int f \, d\mu = \bigvee_{m, n} s_{f, m, n} = 0$.

  \emph{$\omega$-continuity.}
  The integral is defined in terms of the following operations, all of which are $\omega$-continuous: $f \mapsto [f > q]$, $\mu$, addition, scalar multiplication and join.

  \emph{Additivity.}
  Let $f, g : A \rightarrow \mathbb{R}_l^+$.
  Let $m, n \geq 1$.
  Note that $\int f \, d\mu + \int g \, d\mu = \bigvee_{m, n} \, (s_{f, m, n} + s_{g, m, n})$ because the families $(s_{-, m,n })_{mn}$ are directed (Lemma \ref{lem:sfmn-directed}) and addition is $\omega$-continuous.
  Application of Lemma \ref{lem:preimages-of-sums} for the functions $mf$ and $mg$ gives
  \begin{equation}
    s_{f, m, n} + s_{g, m, n} = \frac{1}{m} \sum_{k = 1}^{2mn} \mu(\underbrace{\bigcup \{[f > \frac{i}{m}]_0 \cap [g > \frac{j}{m}]_0 \mid 0 \leq i, j \leq mn, i + j = k \}}_{\subseteq [f + g > \frac{k}{m}]}),
  \end{equation}
  which is $\leq s_{(f + g), m, 2n}$.
  Letting $n$ and $m$ vary, we conclude $\int f \, d\mu + \int g \, d\mu \leq \int (f + g) \, d\mu$.

  On the other hand, let $q \in \mathbb{Q}$ such that $q < \int f + g \, d\mu$.
  We will show $q < \int f \, d\mu + \int g \, d\mu$.
  By definition of $\int - \, d\mu$ as a join, there merely exist $m, n \in \mathbb{N}$ such that $q < s_{f + g, m, n}$.
  Thus there are rational numbers $q_k < \mu([f + g > \frac{k}{m}])$ for $1 \leq k \leq nm$ such that $q = \frac{1}{m} \sum_{k = 1}^{mn} q_k$.
  We have
  \begin{equation}
    [f + g > \frac{k}{m}] = \bigcup_{m | m'} \bigcup \{[f > \frac{i}{m'}] \cap [g > \frac{j}{m'}] \mid i, j \in \mathbb{N}, \frac{i + j}{m'} = \frac{k}{m} \}
  \end{equation}
  for each $k$ and the outer union on the right-hand side is directed, with upper bounds given by common multiples of the $m'$.
  Thus $\mu$ commutes with the outer union.

  It follows that for each $k$ there is $m'_k$ such that
  \begin{equation}
    \label{eq:def-qk}
    q_k < \mu(\bigcup \{[f > \frac{i}{m'_k}] \cap [g > \frac{j}{m'_k}] \mid i,j \in \mathbb{N}, \frac{i + j}{m'_k} = \frac{k}{m}\}).
  \end{equation}
  By taking upper bounds wrt.\@ divisibility, we may assume $m'_k = m'$ for all $k$ and a single $m'$ such that $m | m'$.
  We obtain
  \begin{align}
    & \int f \, d\mu + \int f\, d\mu \\
    \geq{} & s_{f, m', n} + s_{g, m', n}\\
    ={} & \frac{1}{m'} \sum_{\ell = 1}^{2m'n} \mu(\bigcup \{[f > \frac{i}{m'}] \cap [g > \frac{j}{m'}] \mid i, j \in \mathbb{N}, i + j = \ell\}) \\
    ={} & \frac{1}{m} \sum_{\ell' = 0}^{2mn - 1} \frac{m}{m'} \sum_{\ell'' = 1}^{m'/m} \mu(\bigcup \{[f > \frac{i}{m'}] \cap [g > \frac{j}{m'}] \mid i, j \in \mathbb{N}, i + j = \ell' \frac{m'}{m} + \ell'' \}) \\
  \end{align}
  by decomposing the index $\ell$ as $\ell = \ell' \frac{m}{m'} + \ell''$.
  Now
  \begin{align}
    & \bigcup \{[f > \frac{i}{m'}] \cap [g > \frac{j}{m'}] \mid i, j \in \mathbb{N}, i + j = \ell' \frac{m'}{m} + \ell'' \} \\
     \supseteq {} & \bigcup \{[f > \frac{i}{m}] \cap [g > \frac{j}{m}] \mid i,j \in \mathbb{N}, i + j = \ell' + 1\},
  \end{align}
  for all $\ell$ and $\ell'$, which is independent of $\ell'$.
  Thus
  \begin{align}
    & \int f \, d\mu + \int f\, d\mu \\
    \geq {}& \frac{1}{m} \sum_{k = 1}^{2mn} \underbrace{\mu(\bigcup \{[f > \frac{i}{m}] \cap [g > \frac{j}{m}] \mid i,j \in \mathbb{N}, i + j = k \})}_{> q_k \text{ if } k \leq nm} \\
    > {}& \frac{1}{m} \sum_{k = 1}^{mn} q_k \\
    = {}& q
  \end{align}
  where we reindexed with $k = \ell' + 1$ and used equation \eqref{eq:def-qk}.
  $q < \int f + g \, d\mu$ was arbitrary, hence $\int f + g \, d\mu \leq \int f \, d\mu + \int g \, d\mu$.
\end{proof}

\begin{proof}[Proof of Theorem \ref{th:riesz}.]
  By Lemma \ref{lem:restriction-of-integral-is-valuation}, the restriction $\mu_\mathcal{I}$ of an integral $\mathcal{I}$ to indicator functions is a valuation, and by Lemma \ref{lem:mu-int-is-integral} the assignment $f \mapsto \int f \, d\mu$ is an integral for all valuations $\mu$.
  The two functions are monotone and restrict to functions on sub-probability valuations and integrals.
  It remains to show that
  \begin{enumerate}[label=(\arabic*)]
    \item
      \label{itm:integral-extends-measure}
      $\int - \, d\mu$ is an extension of $\mu$, i.e.\@ $\int \mathds{1}_U \, d\mu = \mu(U)$ for all opens $U \in \mathcal{O}(A)$, and
    \item
      \label{itm:integral-uniquely-determined}
      every integral is uniquely determined by its value on indicator functions.
  \end{enumerate}

  \ref{itm:integral-extends-measure}.
  Let $U \in \mathcal{O}(A)$ be an open subset.
  Then $[\mathds{1}_U > q] = \emptyset$ for all $q \geq 1$, and $[\mathds{1}_U > q] = U$ for all $q < 1$.
  Thus
  \begin{equation}
    s_{\mathds{1}_U, m, n} = \frac{1}{m} \sum_{i = 1}^{mn} \mu([\mathds{1}_U > \frac{i}{m}]) = \frac{1}{m} \sum_{i = 1}^{m - 1} \mu(U) = \frac{m - 1}{m} \mu(U)
  \end{equation}
  for all $m, n > 1$, and we conclude $\int \mathds{1}_U \, d\mu = \bigvee_{m > 0} \frac{m - 1}{m} \mu(U) = \mu(U)$.

  \ref{itm:integral-uniquely-determined}.
  Let $\mathcal{I}$ be an integral and let $f : A \rightarrow \mathbb{R}_l^+$.
  Then
  \begin{equation}
    f = \bigvee_{m, n \geq 1} \frac{1}{m} \sum_{i = 1}^{mn} \mathds{1}_{[f > \frac{i}{m}]}.
  \end{equation}
  and this join is directed (for the same reason that $(s_{f, m, n})_{mn}$ is a directed family).
  By linearity (Proposition \ref{prop:integrals-linear}) and $\omega$-continuity of $\mathcal{I}$, we have
  \begin{equation}
    \mathcal{I}(f) = \bigvee_{m, n \geq 1} \frac{1}{m} \sum_{i = 1}^{mn} \mathcal{I}(\mathds{1}_{[f > \frac{i}{m}]}),
  \end{equation}
  thus $\mathcal{I}$ is uniquely determined by its restriction to indicator functions.
\end{proof}

\section{The Giry monad}
\label{sec:giry}

By definition, there are inclusions $\mathfrak{G}_{\leq 1}(A) \subseteq \mathfrak{G}(A) \subseteq ((A \rightarrow \mathbb{R}_l^+) \rightarrow \mathbb{R}_l^+)$.
The operator $\mathrm{Cont}_{\mathbb{R}_l^+} : A \mapsto ((A \rightarrow \mathbb{R}_l^+) \rightarrow \mathbb{R}_l^+)$ is the \emph{continuation monad} \citep{moggi1991notions} instantiated with $\mathbb{R}_l^+$.
As we are working internally (i.e.\@ an internal monad corresponds to an external \emph{strong} monad), monad structure on an operator $M : \mathrm{Set} \rightarrow \mathrm{Set}$ is given by \emph{unit} maps $\eta : A \rightarrow M(A)$ and \emph{bind} maps $>\!\!>\!= : M(A) \times (A \rightarrow M(B)) \rightarrow M(B)$ for all sets $A, B$, which satisfy unit and associativity laws.
In case of the continuation monad $M = \mathrm{Cont}_{\mathbb{R}_l^+}$,
\begin{equation}
  \eta(a) = (f \mapsto f(a))
\end{equation}
is the map that evaluates a given $f : A \rightarrow \mathbb{R}_l^+$ at a certain $a \in A$, and bind is given by
\begin{equation}
  x >\!\!>\!= y = (f \mapsto x (a \mapsto y(a)(f))),
\end{equation}
where $x \in M(A)$, $y : A \rightarrow M(B)$, $f : A \rightarrow \mathbb{R}_l^+$ and $a \in A$.

By the Riesz Theorem~\ref{th:riesz}, $\mathfrak{G}(A) \cong \mathfrak{V}(A)$ and $\mathfrak{G}_{\leq 1}(A) \cong \mathfrak{V}_{\leq 1}(A)$.
This justifies defining the Giry monad of (sub-probability) valuations as follows:

\begin{definition-proposition}[\href{https://github.com/FFaissole/Valuations/blob/d06d2c8c9cce3ddf6137ca3440ab02031912d292/Giry.v}{\tt Giry.v}]
  \label{def-prop:giry}
  The unit and bind operations of the continuation monad $\mathrm{Cont}_{\mathbb{R}_l^+}$ restrict to operations on (sub-probability) integrals.
  The \emph{(sub-probabilistic) Giry monad} is given by the operator $A \mapsto \mathfrak{G}(A)$ (resp.\@ $A \mapsto \mathfrak{G}_{\leq 1}(A)$) and the restricted unit and bind operations of the continuation monad.
\end{definition-proposition}
\begin{proof}
  We need to show stability of $\mathfrak{G}$ and $\mathfrak{G}_{\leq 1}$ under $\eta$ and $>\!\!>\!=$.
  The verification of the rules of lower integrals is done by unfolding the pointwise definition of addition and the partial ordering on functions $A \rightarrow \mathbb{R}_l^+$.
  We show how some of the rules can be derived, the other proofs being similar.

  If $a \in A$ and $f, g : A \rightarrow \mathbb{R}_l^+$, then $\eta(a)(f + g) = (f + g)(a) = f(a) + g(a) = \eta(a)(f) + \eta(a)(g)$, thus $\eta(a)$ is linear.
  Let $\mathcal{I} \in \mathfrak{G}(A)$ and $\mathcal{J} : A \rightarrow \mathfrak{G}(B)$.
  $\omega$-continuity of $\mathcal{I} >\!\!>\!= \mathcal{J}$ can be seen as follows.
  Let $U \subseteq (B \rightarrow \mathbb{R}_l^+)$ be a directed enumerable subset of the function space.
  Then for each $a \in A$ it holds that $\mathcal{J}(a)(\bigvee U) = \bigvee_{f \in U} \mathcal{J}(a)(f)$ because $\mathcal{J}(a)$ is $\omega$-continuous.
  Thus
  \begin{align}
    (\mathcal{I} >\!\!>\!= \mathcal{J})(\bigvee U)
    ={} & \mathcal{I}(a \mapsto \bigvee_{f \in U} \mathcal{J}(a)(f)) \\
    ={} & \mathcal{I}(\bigvee_{f \in U} \, (a \mapsto \mathcal{J}(a)(f))) \\
    ={} & \bigvee_{f \in U} \mathcal{I}(a \mapsto \mathcal{J}(a)(f)) \\
    ={} & \bigvee_{f \in U} \, (\mathcal{I} >\!\!>\!= \mathcal{J})(f)
  \end{align}
  using the pointwise definition of joins on $A \rightarrow \mathbb{R}_l^+$ and the $\omega$-continuity of $\mathcal{I}$.

  We have $\eta(a)(\mathds{1}_A) = \mathds{1}_A(a) = 1$ for all $a \in A$, so $\eta$ is valued in sub-probability integrals.
  If $\mathcal{I} \in \mathfrak{G}_{\leq 1}(A)$ and $\mathcal{J} : A \rightarrow \mathfrak{G}_{\leq 1}(B)$, then $a \mapsto \mathcal{J}(a)(\mathds{1}_B)$ is a function $\leq \mathds{1}_A$ because $\mathcal{J}(a)$ is a sub-probability integral on $B$ for all $a$.
  Thus $(\mathcal{I} >\!\!>\!= \mathcal{J})(\mathds{1}_B) \leq \mathcal{I}(\mathds{1}_A) \leq 1$ by monotonicity of $\mathcal{I}$ and $\mathcal{I}$ being sub-probabilistic.
\end{proof}

\citet{Vickers} proves that the variant of the Giry monad on the category of locales is commutative.
Commutativity of $\mathfrak{G}$ would mean that for $\mathcal{I} \in \mathfrak{G}(A)$ and $\mathcal{J} \in \mathfrak{G}(B)$ the two integrals
\begin{equation}
  \label{eq:giry-strength}
  (\mathcal{I} \blacktriangleright \mathcal{J})(f) = \mathcal{I}(a \mapsto \mathcal{J}(b \mapsto f(a, b)))
\end{equation}
and
\begin{equation}
  (\mathcal{I} \blacktriangleleft \mathcal{J})(f) = \mathcal{J}(b \mapsto \mathcal{I}(a \mapsto f(a, b)))
\end{equation}
on $A \times B$ agree.\footnote{Vickers uses $\triangleleft$ and $\triangleright$ instead of $\blacktriangleleft$ and $\blacktriangleright$ here; we reserve unfilled triangles for the (unrelated) cover relations of $\omega$-cpo presentations.}
In classical mathematics, this corresponds to the Fubini theorem
\begin{equation}
  \int_A \left(\int_B f(a, b) \, db\right) \, da = \int_B \left(\int_A f(a, b) \, da\right) \, db = \int_{A \times B} f(a, b) \, d(a, b)
\end{equation}
and uniqueness of product measures.

However, the proof given in \citet{Vickers} does not directly translate to our setting because it relies on the product $X \times Y$ of locales being dual to the coproduct $\mathcal{O}(X) \otimes \mathcal{O}(Y)$ of underlying frames.
In synthetic topology, this corresponds to products having the product topology:
\begin{definition}
  Let $A$ and $B$ be sets.
  For $U \in \mathcal{O}(A)$ and $V \in \mathcal{O}(B)$, let
  \begin{equation}
    U \times V = \{ (a, b) \in A \times B \mid a \in U \land b \in V \} \in \mathcal{O}(A \times B);
  \end{equation}
  it is open because $\mathbb{S}$ is closed under binary meets.
  $A \times B$ has the \emph{product topology} if $\mathcal{O}(A \times B) \subseteq \mathcal{P}(A \times B)$ is the least subset that is closed under enumerable joins and contains the sets $U \times V$ for all $U \subseteq A$ and $V \subseteq B$ open.
\end{definition}
Note that we require that the topology on $A \times B$ is generated by the basic opens $U \times V$ under enumerable unions, as opposed to arbitrary ones.
Our notion of product topology is in a sense weaker than the one that can be found in \citet[definitions 2.57 and 2.55]{lesnikphd}.
There it is required that every open is an overt (e.g.\@ countable in our case) union of the basic opens $U \times V$, while for us the opens need only be \emph{generated} by basic opens under enumerable unions.
The situation is comparable to the initial $\sigma$-frame and the Rosolini dominance:
In the presence of countable choice, the two definitions are equivalent.

The problem with the Fubini theorem in synthetic topology is that $A \times B$ does not always have the product topology.
Fortunately, $A \times B$ does have the product topology in many special cases. 
Le{\v{s}}nik proves that if $A$ and $B$ are \emph{strongly locally compact}, then $A \times B$ has the product topology (\cite{lesnikphd}, proposition 2.59).
Thus finite products of countable discrete spaces and locally compact metric spaces (e.g.\@ $\mathbb{R}$ under suitable hypotheses, see Section \ref{sec:lebesgue}) behave well, and our Fubini theorem applies.
\begin{theorem}[Fubini]
  \label{th:fubini}
  Let $\mathcal{I} \in \mathfrak{G}(A)$ and $\mathcal{J} \in \mathfrak{G}(B)$ for sets $A, B$.
  Suppose that $A \times B$ has the product topology.
  Then the two integrals $\mathcal{I} \blacktriangleleft \mathcal{J}$ and $\mathcal{I} \blacktriangleright \mathcal{J}$ on $A \times B$ agree.
\end{theorem}
The proof of Theorem \ref{th:fubini} will occupy the remainder of Section \ref{sec:giry}.
It is a direct translation of the proof given by \citet{Vickers} for locales.
\begin{theorem}[Principle of inclusion and exclusion, \cite{Vickers}]
  \label{th:inclusion-exclusion}
  Let $L$ be a lattice with bottom element.
  Then for all $x_1, \dots, x_n \in L$ it holds in $M(L)$ that
  \begin{equation}
    (\bigvee_{i = 1}^n x_i) + \sum_{\substack{I \subseteq \{1, \dots, n\} \\ |I| \text{ is even}}} x_I = \sum_{\substack{|I| \subseteq \{1, \dots, n\} \\ |I| \text{ is odd}}} x_I
  \end{equation}
  where $x_I = \bigwedge_{i \in I} x_i$ for $I \subseteq \{ 1, \dots, n \}$ decidable.
\end{theorem}

\begin{lemma}[\cite{Vickers}]
  \label{lem:Rl_cancellation}
  Let $u_1, \dots, u_n \in \mathbb{R}_l$ and $v$ be lower reals.
  Then the equation $\sum_{i = 1}^n u_i + x = v$ has at most one solution $x$ such that $u_i \leq x$ for all $i$.
\end{lemma}
Note that Lemma \ref{lem:Rl_cancellation}, as stated in the reference refers to the standard lower reals, which are not required to be valued in $\mathbb{S}$.
However, the proof given there also works for our lower reals; moreover, the open lower reals embed into the standard lower reals so that uniqueness for the latter implies uniqueness for the former.

\begin{lemma}
  \label{lem:product-indicator}
  Let $U \in \mathcal{O}(A)$ and $V \in \mathcal{O}(B)$ be opens in sets $A, B$.
  Then for all $a \in A$ and $b \in B$ it holds that
  \begin{equation}
    \mathds{1}_{U \times V}(a, b) = \mathds{1}_U(a) \cdot \mathds{1}_V(b).
  \end{equation}
\end{lemma}
\begin{proof}
  By definition of indicator functions and $U \times V$, the lemma will follow if we can show
  \begin{equation}
    \label{eq:r-meet-vs-product}
    r(s \land t) = r(s) r(t)
  \end{equation}
  for all $s, t \in \mathbb{S}$, where $r : \mathbb{S} \rightarrow \mathbb{R}_l$ is the unique map of $\omega$-cpos satisfying $r(\bot) = 0$ and $r(\top) = 1$.
  For fixed $t$, the two $\omega$-continuous maps $s \mapsto r(s \land t)$ and $s \mapsto r(s) r(t)$ agree for $s = \bot$ and $s = \top$, so they agree by the universal property of $\mathbb{S} = \mathbb{B}_\bot$.
\end{proof}

\begin{proof}[Proof of Theorem \ref{th:fubini}.]
  For $U \in \mathcal{O}(A)$ and $V \in \mathcal{O}(B)$, we compute with Lemma \ref{lem:product-indicator}
  \begin{align}
    (\mathcal{I} \blacktriangleleft \mathcal{J})(\mathds{1}_{U \times V})
    ={} & \mathcal{J}(b \mapsto \mathcal{I}(a \mapsto \mathds{1}_U(a) \mathds{1}_V(b))) \\
    ={} & \mathcal{J}(\mathcal{I}(\mathds{1}_U) \mathds{1}_V) \\
    ={} & \mathcal{I}(\mathds{1}_U) \mathcal{J}(\mathds{1}_V)
  \end{align}
  and hence by symmetry
  \begin{equation}
    (\mathcal{I} \blacktriangleleft \mathcal{J})(\mathds{1}_{U \times V}) = (\mathcal{I} \blacktriangleright \mathcal{J})(\mathds{1}_{U \times V}) = \mathcal{I}(\mathds{1}_U) \mathcal{J}(\mathds{1}_V).
  \end{equation}
  Because integrals are uniquely determined by their restriction to valuations, it will be sufficient to show that a valuation $\mu$ on $A \times B$ is in turn uniquely determined by its restriction to opens of the form $U \times V$.
  $A \times B$ has the product topology, so $\mathcal{O}(A \times B)$ is the least set containing subsets of the form $U \times V$ with $U \subseteq A$ and $V \subseteq B$ open that is closed under enumerable unions.
  Equivalently, $\mathcal{O}(A \times B)$ is generated under \emph{directed} enumerable unions from opens of the form $U_1 \times V_1 \cup \dots \cup U_n \times V_n$ for $U_i \subseteq A$ open and $V_i \subseteq B$ open, $1 \leq i \leq n$.
  It will thus suffice to prove that $\mu$ is uniquely determined by its value on finite unions of products of opens.
  Applying the principle of inclusion and exclusion (Theorem \ref{th:inclusion-exclusion}), we obtain
  \begin{equation}
    \mu(\bigcup_{i = 1}^n \, (U_i \times V_i)) + \sum_{\substack{I \subseteq \{1, \dots, n\} \\ |I| \text{ is even}}} \mu(U_I \times V_I) = \sum_{\substack{|I| \subseteq \{1, \dots, n\} \\ |I| \text{ is odd}}} \mu(U_I \times V_I),
  \end{equation}
  where $U_I = \bigcap_{i \in I} U_i$ and $V_I = \bigcap_{i \in I} V_i$ (hence $\bigcap_{i \in I} \, (U_i \times V_i) = U_I \times V_I$).
  By monotonicity of $\mu$, it holds that $\mu(U_I \times V_I) \leq \mu(\bigcup_{i = 1}^n \, (U_i \times V_i))$, so by Lemma \ref{lem:Rl_cancellation} the values $\mu(U_I \times V_I)$ uniquely determine $\mu(\bigcup_{i = 1}^n \, (U_i \times V_i))$.
\end{proof}

\section{The Lebesgue valuation}
\label{sec:lebesgue}

Having studied valuations in general, we now turn to constructing a concrete valuation on a non-discrete space: The Lebesgue valuation on the reals.
For this we will need that the intrinsic topology of the Dedekind reals agrees with the topology that is induced by the Euclidean metric, i.e.\@ that $\mathbb{R}$ is \emph{metrizable} \citep{lesnikphd}.
We proceed by defining a $\sigma$-frame of formal real opens and state metrizability as an isomorphism between the formal and the intrinsic real opens.
The Lebesgue valuation can then be defined by a universal property.

\begin{definition}
  The partial order $L$ is the least suborder of $\mathcal{P}(\mathbb{Q})$ containing the sets $(a, b) = \{ x \in \mathbb{Q} \mid a < q < b \}$ for all $a, b \in \mathbb{Q}$ that is closed under binary unions.
\end{definition}
Every element $x \in L$ has a unique presentation as a disjoint union
\begin{equation}
  \label{eq:finite-union-of-rational-intervals}
  x = (a_1, b_1) \mathop{\dot\cup} \dots \mathop{\dot\cup} (a_n, b_n)
\end{equation}
for rational numbers $a_i, b_i$ such that $a_i < b_i \leq a_{i + 1}$ for $i = 1, \dots, n - 1$.
We refer to the elements $(a_i, b_i)$ as the \emph{connected components} of $x$.
The decomposition into connected components can be used to construct $L$ as a subset of lists of rational numbers, and this definition is purely combinatorial and does not use the subobject classifier $\Omega$.
It also follows from the decomposition that $L$ is a distributive lattice with bottom element, i.e.\@ that it has meets:
We have
\begin{equation}
  \bigg(\bigcup_{i = 1}^m (a_i, b_i)\bigg) \cap \bigg(\bigcup_{j = 1}^n (a'_j, b'_j)\bigg) = \bigcup_{i, j} \big((a_i, b_i) \cap (c_j, d_j)\big)
\end{equation}
and $(a_i, b_i) \cap (c_j, d_j) = (\mathrm{max}(a_i, c_j), \mathrm{min}(b_i, d_j))$ for all $i, j$, which is in $L$.

\begin{definition}
  \label{def:real-cover-relation}
  The \emph{cover relation on $L$} is generated by
  \begin{equation}
    \label{eq:formal-generating-covers}
    (a, b) \triangleleft \bigg\{ \bigcup_{j = 1}^n \, (a'_j, b'_j) \mid a < a'_j, b'_j < b \text { for } j \leq n \bigg\}
  \end{equation}
  for $a < b$ under binary unions.
\end{definition}
Thus $(\bigcup_{i = 1}^m (a_i, b_i)) \triangleleft U$ if $U$ is the set of elements $\bigcup_{j = 1}^n (a'_j, b'_j)$ such that for each $j$ there exists $i$ with $a_i < a'_j$ and $b'_j < b_i$.
This cover relation is stable under binary meets and, by definition, joins.
It follows that the $\omega$-cpo completion $L_\omega$ has enumerable joins and finite meets satisfying the distributivity law.
The bottom element of $\emptyset \in L$ is also a bottom element of $L_\omega$.
Finally, the subset of elements of $L_\omega$ which are bounded by $t = \bigvee_{n \in \mathbb{N}} (-n, n)$ contains the image of $L$ and is closed under joins, hence $t$ is a top element of $L_\omega$.
Thus $L_\omega$ is a $\sigma$-frame.

\begin{definition}
  \label{def:formal-real-opens}
  The \emph{$\sigma$-frame of formal real opens} $\mathcal{O}(\mathbb{R}_F)$ is given by the $\omega$-cpo completion of $L$ with respect to the covers of Definition \ref{def:real-cover-relation}.
\end{definition}

By definition, $L \subseteq \mathcal{P}(\mathbb{R})$, but in fact the rational intervals $(a, b)$ are open:
Given a Dedekind real $r = (\ell, u) \in \mathbb{R}$, we have $r \in (a, b)$ if and only if $\ell(a) \land u(b)$, which is a truth value in $\mathbb{S}$ because $\ell, u : \mathbb{Q} \rightarrow \mathbb{S}$.
It follows that $L \subseteq \mathcal{O}(\mathbb{R})$.
This inclusion is cover preserving because $(a, b) = \bigcup \{(a', b') \mid a < a' \leq b' < b \}$ as subsets of $\mathbb{R}$.
We obtain a morphism of $\omega$-cpos $\mathcal{O}(\mathbb{R}_F) \rightarrow \mathcal{O}(\mathbb{R})$.
It is not necessarily an isomorphism, but it will be assumed for the remainder of this section that it is:
\begin{assumption}
  \label{ass:reals-metrizable}
  The map $\mathcal{O}(\mathbb{R}_F) \rightarrow \mathcal{O}(\mathbb{R})$ is an isomorphism of partial orders.
\end{assumption}
\citet[Section 5.3]{lesnikphd} proves that if one assumes the intuitionistic principles function-function choice, the continuity principle (which is absurd in classical logic) and the fan principle (every decidable bar is uniform), then every complete metrically separable metric space is metrized.
In particular, every open $U \in \mathcal{O}(\mathbb{R})$ is a countable union of metric balls, from which our Assumption \ref{ass:reals-metrizable} follows.
Lesnik's assumptions hold in the K2 realizability topos and the big topos of topological spaces, so Assumption \ref{ass:reals-metrizable} holds in these models, too.

\begin{definition-proposition}
  The map $\lambda' : L \rightarrow \mathbb{Q}^+$ given by
  \begin{equation}
    \lambda'(\bigcup_{i = 1}^n (a_i, b_i)) = \sum_{i = 1}^n b_i - a_i;
  \end{equation}
  for $n \geq 0$ and rationals $a_i < b_i \leq a_{i + 1}$, $1 \leq i \leq n - 1$, is well-defined, monotone and, when coerced to a function $L \rightarrow \mathbb{R}_l^+$, cover-preserving.
  The induced map $\lambda : \mathcal{O}(\mathbb{R}) \cong L_\omega \rightarrow \mathbb{R}_l^+$ is a valuation, which we refer to as the \emph{Lebesgue valuation}.
\end{definition-proposition}
\begin{proof}
  $\lambda'$ is well-defined because decompositions into connected components are unique up to reordering.
  It is evidently monotone.
  If $(a, b) \triangleleft U$, then $(a + n^{-1}, b - n^{-1}) \in U$ for all $n > 0$, so that
  \begin{equation}
    \bigvee_{u \in U} \lambda'(u) \geq \bigvee_{n > 0} \lambda'((a + n^{-1}, b - n^{-1})) = b - a = \lambda'((a, b)).
  \end{equation}
  It follows that $\lambda'$ preserves general covers because we have $\lambda'(x \cup y) = \lambda'(x) + \lambda'(y)$ if $x$ and $y$ are disjoint.

  $\lambda$ preserves the bottom element because $\lambda'$ does, and it is $\omega$-continuous by construction.
  What remains to be proved is the modular law
  \begin{equation}
    \label{eq:lebesgue-modularity}
    \lambda(x \cup y) + \lambda(x \cap y) = \lambda(x) + \lambda(y)
  \end{equation}
  for all $x, y \in \mathcal{O}(\mathbb{R})$, but we immediately reduce to $x, y \in L$ by induction.
  In turn, we prove equation \eqref{eq:lebesgue-modularity} for $x, y \in L$ by induction over the total number of connected components of $x$ and $y$.
  It holds trivially if $x = \emptyset$ or $y = \emptyset$.
  If $x = (a, b)$ and $y = (c, d)$ are rational intervals, then
  \begin{align}
    & \lambda(x \cup y) + \lambda(x \cap y) \\
    ={}& \max(b, d) - \min(a, c) + \min(b, d) - \max(a, c) \\
    ={}& b + d - a - c \\
    ={}& \lambda(x) + \lambda(y)
  \end{align}
  so the equation holds in this case, too.

  In the induction step we are given disjoint unions $(a, b) \mathop{\dot\cup} x$ and $(c, d) \mathop{\dot\cup} y$ such that $b < r$ for all $r \in x$ and $d < s$ for all $d \in y$, at least after reordering the connected components if necessary.
  If $n$ is the number of connected components of $x$ and $m$ that for $y$, we may assume that \eqref{eq:lebesgue-modularity} holds for all pairs of elements of $L$  whose total number of connected components is at most $n + m + 1$.

  Suppose first that $(a, b)$ and $(c, d)$, are disjoint, wlog.\@ say $b \leq c$.
  Then $(a, b)$ is disjoint from all of $x$, $(c, d)$ and $y$, thus
  \begin{equation}
    \lambda((a, b) \cup x \cup (c, d) \cup y) = \lambda((a, b)) + \lambda(x \cup (c, d) \cup y).
  \end{equation}
  By the induction hypothesis,
  \begin{equation}
    \lambda(x \cup ((c, d) \cup y)) = \lambda(x) + \lambda((c, d) \cup y) - \lambda(x \cap ((c, d) \cup y)).
  \end{equation}
  Because $(a, b)$ is disjoint from $(c, d)$ and $y$, we have
  \begin{equation}
    ((a, b) \cup x) \cap ((c, d) \cup y) = x \cap ((c, d) \cup y),
  \end{equation}
  which combined with the previous equations yields the modular law for $(a, b) \cup x$ and $(c, d) \cup y$ if $(a, b)$ and $(c, d)$ are disjoint.

  Otherwise $(a, b)$ and $(c, d)$ intersect, so that $(a, b) \cup (c, d) = (e, f)$ with $e = \min(a, c)$ and $f = \max(b, d)$.
  Without loss of generality we may assume $f = d$, so that all of $(a, b)$, $(c, d)$ and $(e, f)$ are disjoint from $y$.
  Thus
  \begin{align}
    & \lambda((a, b) \cup x \cup (c, d) \cup y) \\
    ={}& \lambda(x \cup (e, f) \cup y) \\
    ={}& \lambda(x) + \lambda((e, f)) + \lambda(y) - \lambda(x \cap ((e, f) \cup y)
  \end{align}
  by the induction hypothesis for $x$ and $(e, f) \cup y$.
  The base case of two rational intervals was already proved, thus
  \begin{equation}
    \lambda((e, f)) = \lambda((a, b)) + \lambda((c, d)) - \lambda((a, b) \cap (c, d)).
  \end{equation}
  Because $(a, b)$ is disjoint from $y$ and $\lambda'$ maps disjoint unions to sums, we have
  \begin{equation}
    \lambda(((a, b) \cup x) \cap ((c, d) \cup y)) = \lambda((a, b) \cap (c, d)) + \lambda(x \cap ((c, d) \cup y)).
  \end{equation}
  Putting everything together, we obtain
  \begin{align}
    &\lambda((a, b) \cup x \cup (c, d) \cup y) \\
    ={}& \lambda(x) + \lambda((a, b)) + \lambda(y) + \lambda((c, d)) - \lambda(((a, b) \cup x) \cap ((c, d) \cup y))
  \end{align}
  as required.
\end{proof}

We can now define distributions on (subsets of) the real numbers for which there exists a density with respect to the Lebesgue valuation.
For example, the normal distribution $\mathcal{N}(\mu, \sigma^2)$ has density
\begin{equation}
  f(x) = \frac{1}{\sigma \sqrt{2 \pi}} \exp\Big(-\frac{1}{2} (\frac{x - \mu}{\sigma})^2\Big),
\end{equation}
and so $\mathcal{N}(\mu, \sigma^2) \in \mathfrak{G}_{\leq 1}(\mathbb{R})$ can be defined by
\begin{equation}
  \mathcal{N}(\mu, \sigma^2)(U) = \int_U f \, d\lambda = \int \mathds{1}_U f \, d\lambda
\end{equation}
on opens $U : \mathbb{R} \rightarrow \mathbb{S}$.

\section{Interpreting $\mathcal{R}$ml}
\label{sec:interpretation}

The sub-probability Giry monad $\mathfrak{G}_{\leq 1}$ is defined on the cartesian closed category of sets, and the sets of functions $A \rightarrow \mathfrak{G}_{\leq 1}(B)$ with the pointwise ordering form $\omega$-cpos with bottom elements.
Similarly to \citet{ALEA}, we obtain an interpretation of $\mathcal{R}$ml, a call-by-value PCF with recursion and probabilistic choice as effect.
Because $\mathfrak{G}_{\leq 1}$ is defined in terms of the intrinsic topology (as opposed to the discrete one), we obtain furthermore an interpretation of primitives for sampling from continuous distributions.

First recall the language PCF as in e.g.\@ \citet{plotkin2001adequacy}.
We consider the fragment that has as base types $N$ (natural numbers), $B$ (booleans) and $R$ (real numbers), and as type formers finite products and exponentials.
Thus the set of types $\sigma$ is inductively defined by
\begin{equation}
  \sigma \Coloneqq N \mathbin{|} B \mathbin{|} R \mathbin{|} 1 \mathbin{|} \sigma \times \sigma \mathbin{|} \sigma \rightarrow \sigma.
\end{equation}
Terms $M$ are given by
\begin{align}
  M \Coloneqq {} & 0 \mathbin{|} \textbf{zero}(M) \mathbin{|} \textbf{succ}(M) \mathbin{|} \textbf{pred}(M) \mathbin{|} \textbf{nat-to-real}(M) \mathbin{|} \\
  & \textbf{true} \mathbin{|} \textbf{false} \mathbin{|} \textbf{if } M \textbf{ then } M \textbf{ else } M \mathbin{|} \\
  & M < M \mathbin{|} M + M \mathbin{|} M \cdot M \mathbin{|} \mathrm{exp}(M) \mathbin{|} \mathrm{log}(M) \mathbin{|} \mathrm{sin}(M) \mathbin{|} \dots \mathbin{|} \\
  & \mathbf{bernoulli} \mathbin{|} \mathbf{uniform} \mathbin{|} \dots \mathbin{|} \\
  & \mathbf{*} \mathbin{|} \langle M, M \rangle \mathbin{|} \pi_1(M) \mathbin{|} \pi_2(M) \mathbin{|} \\
  & x \mathbin{|} \lambda x: \sigma. M \mathbin{|} M M \mathbin{|} \\
  & \mathbf{rec}(f : \sigma \rightarrow \tau, x: \sigma.M)
\end{align}
Here $f$ and $x$ can be any one of a fixed set of variables.

There are evident typing rules generating the well-typed terms $\Gamma \vdash M : \sigma$ in a context $\Gamma$ and rules generating equalities $\Gamma \vdash M_0 = M_1 : \sigma$.
(The reduction rules need not concern as here as we are interested in a denotational semantics.)
These rules are set up such that $\mathbf{zero}(M) : B$ is well-typed for natural number terms $M : N$ and equal to $\mathbf{true}$ if and only if $M = 0$.
$\mathbf{pred}$ is the predecessor function (with $\mathbf{pred}(0) = 0$), and $\textbf{nat-to-real}$ is the coercion of natural numbers to real numbers.
The operators $<, +, -, \cdot$ and $\mathrm{exp}, \mathrm{log}, \mathrm{sin}$ are defined on terms of type $R$.
The list of operators can be expanded as needed with more functions that are constructively definable.
Some operators such as the $\mathrm{log}$ function are partial; in these cases the semantics of the program is only defined when the operator is applied to a term that is guaranteed to be in the operator's domain.

The terms $\mathbf{bernoulli} : B$ and $\mathbf{uniform} : R$ are well-typed in every context $\Gamma$; they represent sampling from a fair coin flip and the uniform distribution on the unit interval, respectively.
They are typed like constants of type $B$ and $R$, respectively, and every usage of $\mathbf{bernoulli}$ and $\mathbf{uniform}$ samples a fresh value.
To reuse sampled values, one must thus bind them to variables.

Much like the list of real operators, the list of distributions can be extended as needed with more constructively definable distributions, for example the normal distribution.
Alternatively, many distributions can also be constructed in the language itself if their density with respect to one of the built-in distributions is definable.

The typing rule for the $\mathbf{rec}$ operator is as follows:
\begin{equation}
  \inferrule
  {\Gamma, f : \sigma \rightarrow \tau, x: \sigma \vdash M : \tau}
  {\Gamma \vdash \mathbf{rec}(f : \sigma \rightarrow \tau, x: \sigma.M) : \sigma \rightarrow \tau}
\end{equation}
$\mathbf{rec}$ is used for unbounded recursion and hence satisfies the equation
\begin{equation}
  \mathbf{rec}(f : \sigma \rightarrow \tau, x: \sigma.M) = \lambda x: \sigma.M[\mathbf{rec}(f : \sigma \rightarrow \tau, x: \sigma.M)/f].
\end{equation}

\citet{plotkin2001adequacy} show that cartesian closed categories $\mathcal{C}$ equipped with a suitable monad $T$ can serve as models for call-by-value PCF.
In our case, $\mathcal{C} = \mathrm{Set}$ is the category of sets of our ambient constructive logic, and our monad is $T = \mathfrak{G}_{\leq 1}$, the sub-probabilistic Giry monad.
Let us verify that $\mathrm{Set}$ and $\mathfrak{G}_{\leq 1}$ satisfy the conditions of \citet{plotkin2001adequacy}:
\begin{itemize}
  \item
    $\mathrm{Set}$ is a cartesian closed category enriched over $\omega$-cpos by considering the hom sets as discrete partial orders.
    It has coproducts (hence an object of booleans) and a natural numbers object.
  \item
    The Kleisli category of $\mathfrak{G}_{\leq 1}$ is enriched over $\omega$-cpos with bottom element via the pointwise ordering on maps $A \rightarrow \mathfrak{G}_{\leq 1}(B)$.
    Note that here it is crucial that we work with $\mathfrak{G}_{\leq 1}$ instead of $\mathfrak{G}_{= 1}$ as for the latter the hom sets need not have bottom elements.
  \item
    The strength of $\mathfrak{G}_{\leq 1}$ preserves bottom elements because the integral $\mathcal{I} \mathbin{\blacktriangleright} \mathcal{J}$ of equation \eqref{eq:giry-strength} vanishes if either $\mathcal{I}$ or $\mathcal{J}$ vanishes.
\end{itemize}
Types are now interpreted as follows: 
\begin{mathpar}
  \llbracket N \rrbracket = \mathbb{N} \and
  \llbracket B \rrbracket = \{ 0, 1 \} \and
  \llbracket R \rrbracket = \mathbb{R} \and
  \llbracket 1 \rrbracket = 1 \\
  \llbracket \sigma \times \tau \rrbracket = \llbracket \sigma \rrbracket \times \llbracket \tau \rrbracket \and
  \llbracket \sigma \rightarrow \tau \rrbracket = (\llbracket \sigma \rrbracket \rightarrow \mathfrak{G}_{\leq 1}(\llbracket \tau \rrbracket))
\end{mathpar}
Contexts $\Gamma = (x_1 : \sigma_1, \dots, x_n : \sigma_n)$ are interpreted as products $\llbracket \sigma_1 \rrbracket \times \dots \times \llbracket \sigma_n \rrbracket$.

The denotation of a term $\Gamma \vdash M : \sigma$ is a Kleisli arrow $\llbracket M \rrbracket : \llbracket \Gamma \rrbracket \rightarrow \mathfrak{G}_{\leq 1}(\llbracket \sigma \rrbracket)$; the complete set of clauses can be found in \citet{plotkin2001adequacy}.
For example, the denotation of a successor term $\mathbf{succ}(M)$ for some $\Gamma \vdash M : N$ is defined as composition
\begin{equation}
  \begin{tikzcd}[column sep=huge]
    \llbracket \Gamma \rrbracket \arrow[r, "{\llbracket M \rrbracket}"] & \mathfrak{G}_{\leq 1}(\mathbb{N}) \arrow[r, "\mathfrak{G}_{\leq 1}(\mathrm{succ})"] & \mathfrak{G}_{\leq 1}(\mathbb{N})
  \end{tikzcd}
\end{equation}
using the semantic successor function on the natural numbers object, and the tuple term $\langle M_1, M_2 \rangle$ for terms $\Gamma \vdash M_i : \sigma_i$ is defined as
\begin{equation}
  \begin{tikzcd}[column sep=huge]
    \llbracket \Gamma \rrbracket \arrow[r, "{\langle \llbracket M_1 \rrbracket , \llbracket M_2 \rrbracket \rangle}"] & \mathfrak{G}_{\leq 1}(\llbracket \sigma_1 \rrbracket) \times \mathfrak{G}_{\leq 1}(\llbracket \sigma_2) \rrbracket \arrow[r, "- \mathbin{\blacktriangleright} -"] & \mathfrak{G}_{\leq 1}(\llbracket \sigma_1 \times \sigma_2 \rrbracket)
  \end{tikzcd}
\end{equation}
using the strength of $\mathfrak{G}_{\leq 1}$.

A recursion term $\Gamma \vdash \mathbf{rec}(f : \sigma \rightarrow \tau, x: \sigma.M) : \sigma \rightarrow \tau$ for $\Gamma, f : \sigma \rightarrow \tau, x : \sigma \vdash M : \tau$ is interpreted as follows.
First we define a monotone endofunction $Y$ on the partial order of maps $k : \llbracket \Gamma, x : \sigma \rrbracket \rightarrow \mathfrak{G}_{\leq 1}(\llbracket \tau \rrbracket)$ as follows.
$k$ corresponds to a map $\bar k : \llbracket \Gamma \rrbracket \rightarrow (\llbracket \sigma \rrbracket \rightarrow \mathfrak{G}_{\leq 1}(\llbracket \tau \rrbracket))$.
Now
\begin{equation}
  Y(k) :
  \begin{tikzcd}[column sep=huge]
    \llbracket \Gamma, v : \sigma \rrbracket \arrow[r, "{\langle \pi_1, \bar k \circ \pi_1, \pi_2 \rangle}"]
    & \llbracket \Gamma, f : \sigma \rightarrow \tau, v : \sigma \rrbracket \arrow[r, "\llbracket M \rrbracket"]
    & \mathfrak{G}_{\leq 1}(\llbracket \tau \rrbracket).
  \end{tikzcd}
\end{equation}
We then define $\llbracket \mathbf{rec}(f : \sigma \rightarrow \tau, x: \sigma.M) \rrbracket$ as the join of the sequence $\bot \leq Y(\bot) \leq Y(Y(\bot)) \leq \dots$ of functions $\llbracket \Gamma, v: \sigma \rrbracket \rightarrow \mathfrak{G}_{\leq 1}(\llbracket \tau \rrbracket)$.

The denotations of real functions operators are the corresponding functions on $\mathbb{R}$.
The denotation of the comparison operator $<$ is an exception:
There is no constructively definable function $\mathbb{R} \times \mathbb{R} \rightarrow \{0, 1\}$ that decides the ordering of real numbers; in fact such a function would contradict Assumption \ref{ass:reals-metrizable}.
However, there is a function $l : \mathbb{R} \times \mathbb{R} \rightarrow \mathfrak{G}_{\leq 1}(\{0, 1\})$ such that
\begin{equation}
  l(x_0, x_1) =
    \begin{cases}
      \eta(1) & \text{ if } x_0 < x_1 \\
      \eta(0) & \text{ if } x_0 > x_1 \\
      \bot & \text{ if } x_0 = x_1
    \end{cases}
\end{equation}
and we take $l$ to be the denotation of the $<$ operator.
$l$ is can be constructed as follows:
As proved by \citet{Gilbert}, there is a map
\begin{equation}
  d : \{ (a, b) \in \mathbb{S} \times \mathbb{S} \mid a \land b = \bot \} \rightarrow \{0, 1\}_\bot
\end{equation}
such that $d((\top, \bot)) = 0$ and $d((\bot, \top)) = 1$.
Recall that a Dedekind real number $x$ is given by a pair $(L, U)$ of open sets of rational numbers $L, U : \mathbb{Q} \rightarrow \mathbb{S}$ satisfying conditions such that $L = \{ q \in \mathbb{Q} \mid q < x \}$ and $U = \{ q \in \mathbb{Q} \mid q > x \}$.
Thus given two real numbers $x_1 = (L_1, U_1)$ an $x_2 = (L_2, U_2)$, it holds that $L_1 \cap U_2$ is inhabited if and only if $x_1 > x_2$, and symmetrically that $L_2 \cap U_1$ is inhabited if and only if $x_1 < x_2$.
Inhabitation of open subsets $U \subseteq A$ of countable sets $A$ is an open proposition because it is equivalent to $\bigvee_{a \in A} U(a)$, which is a countable disjunction of open propositions.
It follows that inhabitation of $L_1 \cap U_2$ is a proposition $s_1 \in \mathbb{S}$, and similarly inhabitation of $L_2 \cap U_1$ is a proposition $s_2 \in \mathbb{S}$.
$s_1 \land s_2$ is contradictory.
We now set $l(x_1, x_2) = \phi(d(s_1, s_2))$, where $\phi$ is the canonical map $\{0, 1\}_\bot \rightarrow \mathfrak{G}_{\leq 1}(\{0, 1\})$.

The Bernoulli sampling constant $\Gamma \vdash \mathbf{bernoulli} : B$ is interpreted as the constant function $\llbracket \Gamma \rrbracket \rightarrow \mathfrak{G}_{\leq 1}(\llbracket B \rrbracket)$ with value the lower integral corresponding to the valuation $\nu(0) = \frac{1}{2}$ and $\nu(1) = \frac{1}{2}$.
Similarly, the $\mathbf{uniform}$ term is interpreted as the lower integral corresponding to $\lambda_{(0, 1)}(U) = \lambda(U \cap (0, 1))$ using the Lebesgue valuation of Section \ref{sec:lebesgue}.
(Note that we have to use the open interval $(0, 1)$ here instead of the closed interval $[0, 1]$ because $U \cap [0, 1]$ is not always open in $\mathbb{R}$, even if $U$ is open.)

Consider the problem of sampling from a standard normal distribution given only the $\mathbf{uniform}$ built-in.
One way to achieve this using the \emph{Marsaglia polar method}~\cite{Marsaglia} is as follows:
\begin{lstlisting}[language=ml]
  let rec normal =
    let
      x = 2 * uniform - 1;
      y = 2 * uniform - 1;
      s = x * x + y * y;
    in
      if s < 1 then x * sqrt ((-2) * ln s / s)
      else normal
\end{lstlisting}
Note that \texttt{normal} is defined by potentially infinite recursion but terminates with probability $1$, that is, $\llbracket \texttt{normal} \rrbracket(\mathbb{R}) = 1 \in \mathbb{R}_l$.

\section{Conclusion}
\label{sec:conclusion}

\emph{Contributions.}
This paper develops the foundations of integration theory in synthetic topology based on the initial $\sigma$-frame.
The initial $\sigma$-frame $\mathbb{S}$ is the $\omega$-cpo completion of the booleans.
We discuss several alternative constructions of free $\omega$-cpo completions and show how product $\omega$-cpos can be presented in terms of presentations of their factors.
It is shown that $\mathbb{S}$ is a dominance and hence suitable for synthetic topology.
Following \citet{EscardoKnapp} we show that the $\mathbb{S}$-partial map classifier of a set $A$ is given by its pointed $\omega$-cpo completion $A_\bot$.
A set of lower real numbers based on $\mathbb{S}$ is defined and shown to satisfy the universal property of the $\omega$-cpo completion of the rationals.
This set of lower reals is then used in definitions of valuations and lower integrals which take into account the intrinsic topology induced by $\mathbb{S}$.
The Riesz theorem relating valuations and lower integrals is proved and used to define the Giry monad.
The Fubini theorem is shown to hold for sets $A, B$ whose product has the product topology.
Finally, the Lebesgue valuation is defined under the assumption of metrizability of $\mathbb{R}$, which would be impossible if our valuations were based on discrete topologies.

\emph{Related work.}
Much of our approach to lower integrals is adapted from Steven Vickers's work \citep{Vickers,vickers2008localic} work on the same subject, but in the setting of synthetic topology instead of locale theory.
Lower integrals are better behaved on locales than in synthetic topology in certain aspects.
For example, the Fubini theorem holds without restriction for locales, making the Giry monad commutative, whereas we can only prove the Fubini theorem in synthetic topology on the assumption that the involved products are topologized correctly.
On the other hand, the category of locales is not cartesian closed, whereas the ambient category of sets in synthetic topology is even a elementary topos (or, predicatively, a $\Pi W$-pretopos).

\citet[Section 11]{Shulman:BFT} proves the Brouwer fixpoint theorem in homotopy
type theory using synthetic topology.
He uses modalities to mediate between the homotopical and topological circle and other spaces.
This spatial (modal) type theory is modelled in any local topos, for example Johnstone's topological topos \citep{johnstone1979topological}.
Fourman's big topos that models the intuitionistic principles outlined in Section \ref{sec:lebesgue} is also local.
This paper does not focus on homotopy theory, thus the methodology is different.

\citet{EscardoXu} use a similar big topos, but
restricted to compact spaces to model the fan-theorem in a simple type
theory. \citet{CoquandMannaa} provide a stack
model over Cantor space for univalent type theory. It is likely that 
our work model can be given a constructive treatment by these methods;
see \citet{Coquand:Birmingham}.


There is an interesting analogy with the semantics for higher order probabilistic programming described in \citet{Staton2016SemanticsFP,HeunenKSY17}.
Noting that the category of standard Borel spaces is not Cartesian closed, they embed it into a supercategory (of quasi-Borel spaces) which is closed under exponentials.
A similar problem exists in synthetic topology: The category of topological spaces is not
Cartesian closed. The common solution is to consider a convenient super-category.
\citet[Chapter 10]{Escard2004SyntheticTO} presents a number of subcategories of presheaves over the category of topological spaces
for this purpose. In our case, it is more natural to consider the sheaves for the
open cover topology.
In this light, one could consider our construction as first
embedding in a bigger category with (dependent) function types and then defining the monad on the bigger category.
One advantage of semantics in toposes is that they model all of constructive mathematics, including the principle of unique choice.
This enables use of a strong internal logic to simplify arguments, as is exemplified in this paper.
On the other hand, our Fubini theorem holds only conditionally, whereas it holds for arbitrary products of quasi-Borel spaces, making the Giry monad on quasi-Borel spaces commutative.

\emph{Acknowledgements.}
The questions in this paper originated from discussions with
Christine Paulin in 2014, when Spitters held a Digiteo chair at LRI, Inria. We
also benefited from Faissole’s internship with Paulin about lower reals in Coq.
We are grateful for both.
We thank Alex Kavvos for discussions on the interpretation of effectful programming and fixed points.

This research was partially supported by the Guarded homotopy type theory project, funded by the Villum Foundation, project number 12386, AFOSR project `Homotopy Type Theory and Probabilistic Computation', 12595060, and Digiteo.

\bibliographystyle{abbrvnat}
\setcitestyle{authoryear,open={(},close={)}}
\bibliography{main}

\end{document}